\newcommand{\CLASSINPUTbaselinestretch}{1} 
\newcommand{\CLASSINPUTinnersidemargin}{0.5in} 
\newcommand{\CLASSINPUToutersidemargin}{0.5in} 
\newcommand{\CLASSINPUTtoptextmargin}{0.5in}   
\newcommand{\CLASSINPUTbottomtextmargin}{0.5in}

\documentclass[10pt,journal,doublecolumn,compsoc]{IEEEtran}

\ifCLASSOPTIONcompsoc
  \usepackage[nocompress]{cite}
\else
  \usepackage{cite}
\fi
\ifCLASSINFOpdf
\else
\fi
\usepackage{graphicx,semtrans,amsmath,amsfonts,amssymb,bm,hyperref,url,breakurl,epsfig,epsf,color,MnSymbol,mathbbol,fmtcount,algorithmic,algorithm,caption,subcaption,
  multirow,verbatim}

\usepackage{caption}
\usepackage{cite}

\usepackage[bottom,hang,flushmargin]{footmisc}
\usepackage{comment}
\usepackage{cite}
\usepackage{hyperref}

\usepackage{amssymb,amsfonts,amsmath}

\def\barp{\bar{p}}

\def\by{\mathbf{y}}

\def\bx{\mathbf{x}}

\def\bz{\mathbf{z}}

\def\cN{\mathcal{N}}
\def\cM{\mathcal{M}}

\def\cS{\mathcal{S}}

\def\cO{\mathcal{O}}
\def\cL{\mathcal{L}}
\def\cH{\mathcal{H}}
\def\cT{\mathcal{T}}
\def\bbE{\mathbb{E}}

\def\bbI{\mathbb{1}}

\def\hby{\hat{\by}}
\def\hy{\hat{y}}

\def\cpij{\mathcal{P}_{i\to j}}

\def\Fij{F_{\cpij}}
\def\Fijr{F_{\cpij^r}}
\def\Fijo{F_{\cpij^1}}
\def\Fijt{F_{\cpij^2}}

\def\Eij{E_{i\to j}}
\def\Fbar{\bar{F}}

\def\VI{V}
\def\EI{E}
\def\GI{G}
\def\WI{W}

\def\Vc{V_{c}}
\def\Ec{E_{c}}
\def\Gc{G_{c}}
\def\Wc{W_{c}}

\def\eps{\epsilon}

\def \endprf{\hfill {\vrule height6pt width6pt depth0pt}\medskip}
\newenvironment{proof}{\noindent {\bf Proof} }{\endprf\par}

\newtheorem{theorem}{\textbf{Theorem}}
\newtheorem{definition}{\textbf{Definition}}
\newtheorem{lemma}{\textbf{Lemma}}
\newtheorem{alg}{\textbf{Algorithm}}
\newtheorem{corollary}{\textbf{Corollary}}
\newtheorem{proposition}{\textbf{Proposition}}
\newtheorem{remark}{\textbf{Remark}}
\newtheorem{example}{\textbf{Example}}

\begin{document}
\title{Network Infusion to Infer Information Sources in Networks}
\author{Soheil~Feizi,
        Muriel~M\'edard,
        Gerald~Quon,
        Manolis~Kellis,
        and~Ken~Duffy
\IEEEcompsocitemizethanks{\IEEEcompsocthanksitem S. Feizi, M. M\'edard and M. Kellis are with the Department of Electrical
    Engineering and Computer Science, MIT. G. Quon is with the University of California, Davis. K. Duffy is with the Hamilton Institute, Maynooth University. \protect\\}
}

\IEEEtitleabstractindextext{%
\begin{abstract}
Several significant models have been developed that enable the study of diffusion of signals across biological, social and engineered networks. Within these established frameworks, the inverse problem of identifying the source of the propagated signal is challenging, owing to the numerous alternative possibilities for signal progression through the network. In real world networks, the challenge of determining sources is compounded as the true propagation dynamics are typically unknown, and when they have been directly measured, they rarely conform to the assumptions of any of the well-studied models. In this paper we introduce a method called Network Infusion (NI) that has been designed to circumvent these issues, making source inference practical for large, complex real world networks. The key idea is that to infer the source node in the network, full characterization of diffusion dynamics, in many cases, may not be necessary. This objective is achieved by creating a diffusion kernel that well-approximates standard diffusion models, but lends itself to inversion, by design, via likelihood maximization or error minimization. We apply NI for both single-source and multi-source diffusion, for both single-snapshot and multi-snapshot observations, and for both homogeneous and heterogeneous diffusion setups. We prove the mean-field optimality of NI for different scenarios, and demonstrate its effectiveness over several synthetic networks. Moreover, we apply NI to a real-data application, identifying news sources in the Digg social network, and demonstrate the effectiveness of NI compared to existing methods. Finally, we propose an integrative source inference framework that combines NI with a distance centrality-based method, which leads to a robust performance in cases where the underlying dynamics are unknown.
\end{abstract}

}

\maketitle

\IEEEdisplaynontitleabstractindextext
\IEEEpeerreviewmaketitle
\ifCLASSOPTIONcompsoc
\IEEEraisesectionheading{\section{Introduction}\label{sec:intro}}
\else
\section{Introduction}
\label{sec:intro}
\fi

\IEEEPARstart{I}{nformation} from a single node (entity) can reach other nodes (entities) by propagation over network connections. For
instance, a virus infection (either computer or biological) can propagate to different nodes in a network and become an epidemic \cite{scale-free-epidemics}, while rumors can spread in a social network through social interactions \cite{social-spread-daron}. Even a financial failure of an institution can have cascading effects on other financial entities and may lead to a financial crisis \cite{finance-cascade}. As a final example, in some human diseases, abnormal activities of few genes can cause their target genes and therefore some essential biological processes to fail to operate normally in the cell \cite{disease1,disease2}.

In order to gain insight into these processes, mathematical models have been developed, primarily focusing on application to the study of virus propagation in networks (\cite{valente1996network,achlioptas2009explosive}). A well-established continuous-time diffusion model for viral epidemics is known as the susceptible-infected (SI) model \cite{SIR}, where infected nodes spread the virus to their neighbors probabilistically. For that diffusion model, \cite{newman-epidemic2002spread,newman-percolation-epidemics,ganesh2005effect,scale-free-epidemics} explore the relationship between network structure, infection rate, and the size of the epidemics, while \cite{demiris2005bayesian,tutorial-epidemic2002tutorial,okamura2007statistical} consider learning SI model parameters. Other diffusion methods use random walks to model information spread and label propagation in networks \cite{mostafavi2012labeling,Label-propagation,newman2003mixing}. These references study the forward problem of signal diffusion.

Source inference is the inverse problem. It aims to infer {\it source nodes} in a network by merely knowing the network structure and observing the information spread at single or multiple snapshots (Figure \ref{fig:framework}). Even within the context of the well-studied diffusion kernels, source inference is a difficult problem in great part owing to the presence of path multiplicity in the network \cite{newman2001scientific}. Recently, the inverse problem of a diffusion process in a network under a discrete time memoryless diffusion model, and when time steps are known, has been considered \cite{discrete-diffusion}, while the problem of identifying seed nodes (effectors) of a partially activated network in the steady-state of an Independent-Cascade model is investigated in \cite{Label-propagation}. Moreover reference \cite{farajtabar2015back} has considered the source inference problem using incomplete diffusion traces by maximizing the likelihood of the trace under the learned model. The problem setup and diffusion models considered in those works are different than the continuous-time diffusion setup considered in the present paper. The work in \cite{prakash2012spotting} uses the Minimum Description Length principle to identify source nodes using a heuristic cost function which combines the model cost and the data cost. Moreover, for the case of having a single source node in the network, some methods infer the source node based on distance centrality \cite{newman2010networks}, or degree centrality \cite{book-graphtheory-west2001introduction} measures of the infected subgraph. These methods can be efficiently applied to large networks, but, amongst other drawbacks, their performance lacks provable guarantees in general. For tree structures under a homogeneous SI diffusion model, \cite{shah2011rumors} computes a maximum likelihood solution for the source inference problem and provides provable guarantees for its performance. Over tree structures, their solution is equivalent to the distance centrality of the infected subgraph. The problem of inferring multiple sources in the network has an additional combinatorial complexity compared to the single-source case (see Remark \ref{remark:multi-source}). Reference \cite{lappas2010finding} has considered this problem under an Independent Cascade (IC) model and introduced a polynomial time algorithm based on dynamic-programming for some special cases. We study this problem for incoherent sources and under the SI model in Section \ref{subsec:NI-multiple-source}. We review this prior work more extensively in Section \ref{subsec:prior-work}.

\begin{figure}[t]
  \centering
      \includegraphics[width=0.5\textwidth]{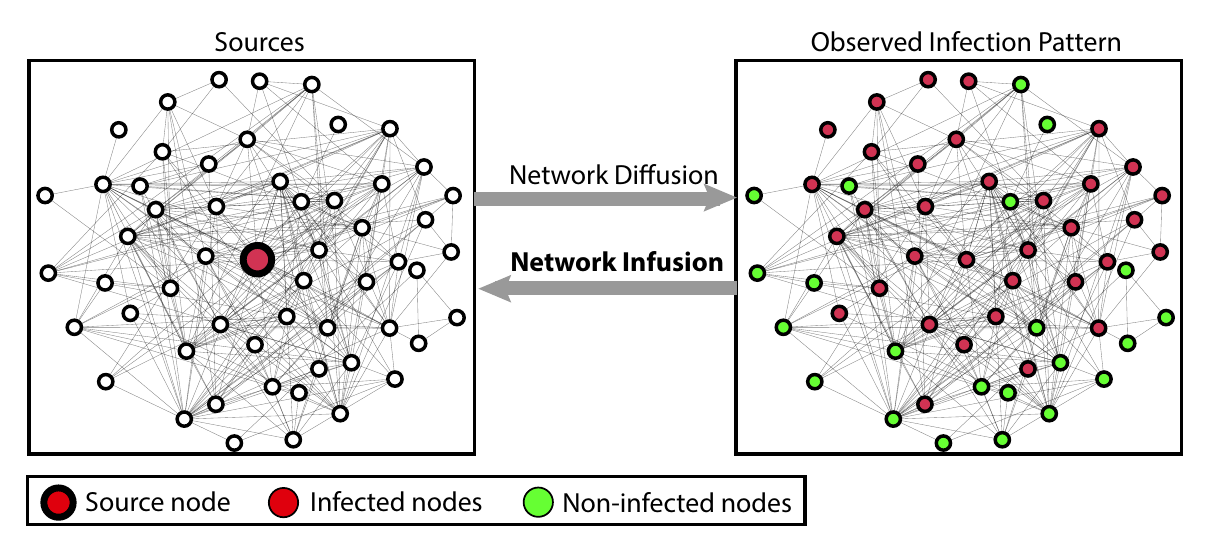}
  \caption{Network Infusion (NI) aims to identify source node(s) by reversing information propagation in the network. NI is based on a path-based network diffusion process that closely approximates the observed diffusion pattern, while leading to a tractable source inference method for large complex networks. The displayed network and infection pattern are parts of the Digg social news network.}
  \label{fig:framework}
\end{figure}

Source inference in real-world networks is made more challenging as the true propagation dynamics are typically unknown \cite{du2013uncover}, and when they have been directly measured, they rarely conform to the assumptions of well-studied canonical diffusion models such as SI, partially owing to heterogeneous information diffusion over network edges, latent sources of information, noisy network connections, non-memoryless transport and so forth \cite{goel2012structure}. As an example, we consider news spread over the Digg social news networks \cite{digg} for more than 3,500 news stories. We find that in approximately $65\%$ of cases, nodes who have received the news at a time $t$, did not have any neighbors who had already received the news by that time violating the most basic conditional independency assumption of the SI model. Furthermore, the empirical distribution of remaining news propagation times over edges of the Digg social news network cannot be approximated closely by a single distribution of a homogenous SI diffusion model, even by fitting a general Weibull distribution to the observed data (Appendix Figure \ref{fig:SIR-digg}).

Owing to high computational complexity of solving the source inference problem under the well-studied SI diffusion models and considering the fact that those kernels are unlikely to match precisely a real-world diffusion, our key idea to solve the inverse problem is to identify a diffusion process that closely approximates the observed diffusion pattern, but also leads to a tractable source inference method by design. Thus, we develop a diffusion kernel that is distinct from the standard SI diffusion models, but its order of diffusion well-approximates many of them in various setups. We shall show that this kernel leads to an efficient source inference method that can be computed efficiently for large complex networks and shall provide theoretical performance guarantees under general conditions. The key original observation, from both a theoretical and practical perspective, is that in order to solve the inverse problem one does not need to know the full dynamics of the diffusion, , instead to solve the inversion one can do so from statistics that are consistent across many diffusion models.

Instead of the full network, our proposed continuous-time network diffusion model considers $k$ edge-disjoint shortest paths among pairs of nodes, neglecting other paths in the network. We call this kernel a {\it path-based network diffusion kernel}. Propagation times in the kernel are stochastically independent, which leads to efficient kernel computation even for large complex networks. Using the path-based network diffusion kernel, we propose a computationally tractable general method for source inference called {\it Network Infusion} (NI), by maximizing the likelihood (NI-ML) or minimizing the prediction error (NI-ME). The NI-ME algorithm is based on an asymmetric Hamming premetric function, and unlike NI-ML, it can be tuned to balance between false positive and false negative error types. Our approach can be computed efficiently for large complex networks, similarly to the distance and degree centrality methods. However, unlike those methods, we provide provable performance guarantees under a continuous-time dynamic SI diffusion setup. We prove that under the SI diffusion model, (i) the maximum-likelihood NI algorithm is mean-field optimal for tree structures, and (ii) the minimum-error NI algorithm is mean-field optimal for regular tree structures. All proofs are presented in the Appendix \ref{sec:proofs}.

Most existing source inference methods consider the case of a single-source homogeneous diffusion setup, partially owing to additional combinatorial complexity of more general cases. We show that the proposed NI framework can be used efficiently for multi-source and heterogeneous diffusion setups under general conditions. Particularly, we prove that the proposed multi-source NI algorithm, which is based on localized likelihood optimization, is mean-field optimal in the regular tree structure for sufficiently-distant sources. Moreover, for a heterogeneous diffusion setup, we show that the network diffusion kernel can be characterized using the phase-type distribution of a Markov chain absorbing time, leading to efficient source inference methods. We also extend our NI algorithms to the cases with unknown or partially known diffusion model parameters such as observation times, by introducing techniques to learn these parameters from observed sample values.

We apply NI to several synthetic networks considering an underlying standard SI diffusion model and compare its performance to existing source inference methods. Having verified the effectiveness of NI both theoretically and through simulations, we then apply it to a real data application to identify the news sources for over 3,500 stories in the Digg social news network. We demonstrate the superior performance of NI compared to existing methods by validating the results based on annotated information, that was not provided to the algorithms. Finally, we propose an integrative source inference framework which combines source prediction ranks of network infusion and distance centrality and leads to a robust performance in cases where the underlying dynamics are unknown.

\section{Problem Setup, Computational Difficulties, and Prior Work}\label{sec:setup}
In this section, we present the source inference problem and explain its underlying challenges. We also review prior work and present notation used in the rest of the paper.

\subsection{Source Inference Problem Setup}\label{subsec:setup}
Let $G=(V,E)$ be a binary, possibly directed, graph representing a network with $n$ nodes, where $G(i,j)=1$ means that there is an edge from node $i$ to node $j$ (i.e., $(i,j)\in E$) and $G(i,j)=0$ means there is none. Let $\cN(i)$ represent the set
of neighbors of node $i$ in the network. If $G$ is directed, $\cN(i)$ represents the set of parents of node $i$. For the sake of description, we illustrate the problem setup and notation in the context of a virus infection spread in the network with the understanding that our framework can be used to solve a more general source inference problem. Let $\cS\subset V$ be the set of source nodes in the network. When a node gets infected, it
spreads infection to its neighbors, which causes the propagation of infection in the network. Let $\cT_{(i,j)}$ be the non-negative, continuous random variable representing the virus traveling time over the edge $(i,j)\in E$. The
$\cT_{(i,j)}$ variables are assumed to be mutually independent. Let $\cpij$ denote a path, i.e. an ordered set of edges, connecting node $i$ to node $j$ in the network. We define $\cT_{\cpij}$ as a random variable representing the virus traveling time over the path $\cpij$, with the following cumulative density function,
\begin{align}\label{eq:cdf-general-path}
\Fij(t)\triangleq Pr[\cT_{\cpij}\leq t].
\end{align}

Let $\by(t)\in\{0,1\}^{n}$ be the node infection vector at time $t$, where $y_i(t)=1$ means that node $i$ is infected at time $t$. Suppose $\cT_i$ is a random variable representing the time that node $i$ gets infected. We assume that if a node gets infected it remains infected (i.e., there is no recovery). Suppose $\tau_i$ is a realization of the random variable $\cT_i$. Thus, $y_i(t)=1$ if $t\geq \tau_i$, otherwise $y_i(t)=0$. If $i$ is a source node, $\cT_i=0$ and $y_i(t)=1$ for all $t\geq 0$. The set $V^{t}=\{i:y_i(t)=1\}$ represents all nodes that are infected at time $t$.

\begin{definition}\label{def:SI-model}
In a dynamic Susceptible-Infected (SI) diffusion setup, we have
\begin{align}
\cT_i= \min_{j\in\cN(i)} (\cT_j+\cT_{(j,i)}).
\end{align}
\end{definition}

Let $\{\by(t): t\in (0,\infty)\}$ represent a continuous-time stationary stochastic process of diffusion in the network $G$. In the source inference problem, given the sample values at times $\{t_1,\ldots,t_z\}$ (i.e., $\{\by(t_1),\ldots,\by(t_z)\}$), as
well as the underlying graph structure $G$, we wish to infer sources nodes that started the infection at time $0$. We assume that the number of sources to be inferred (i.e., $m$) and the observation time stamps (i.e., $\{t_1,\ldots,t_z\}$) are also given. We discuss the cases with unknown or partially known parameters in Section \ref{subsec:parameters}.

One way to formulate the source inference problem is to use a standard maximum a posteriori (MAP) estimation.

\begin{definition}[MAP Source Inference]\label{def:MAP-setup}
The MAP source inference solves the following optimization:

\begin{align}\label{opt:map-setup}
\arg\max \quad& Pr\big(\by(0)|\by(t_1),\ldots,\by(t_z)\big),\\
&\|\by(0)\|_{l_0}=m, \nonumber
\end{align}
\noindent
where $m$ is the number of source nodes in the network, and $\|.\|_{l_0}$ represents the $l_0$ norm of a vector.
\end{definition}
\noindent
In some applications, there may be nonuniform prior probabilities for different candidate source nodes. The MAP source inference optimization takes into account these prior probabilities as well. If there is no informative prior probabilities for candidate source nodes, the MAP source Optimization \eqref{opt:map-setup} can be simplified to the following maximum likelihood (ML) source estimation:

\begin{definition}[ML Source Inference]\label{def:ML-setup}
The ML source inference solves the following optimization:

\begin{align}\label{opt:ml-setup}
\arg\max \quad& Pr\big(\by(t_1),\ldots,\by(t_z)|\by(0)\big),\\
&\|\by(0)\|_{l_0}=m, \nonumber
\end{align}
\noindent
where its objective function is an ML function (score) of source candidates.
\end{definition}
\noindent
An alternative formulation for the source inference problem is based on minimizing the prediction error. In Section \ref{subsec:NI-global}, we propose a minimum prediction error formulation that uses an asymmetric Hamming pre-metric function and can balance between false positive and false negative error types by tuning a parameter.

\begin{figure}[t]
  \centering
      \includegraphics[width=0.4\textwidth]{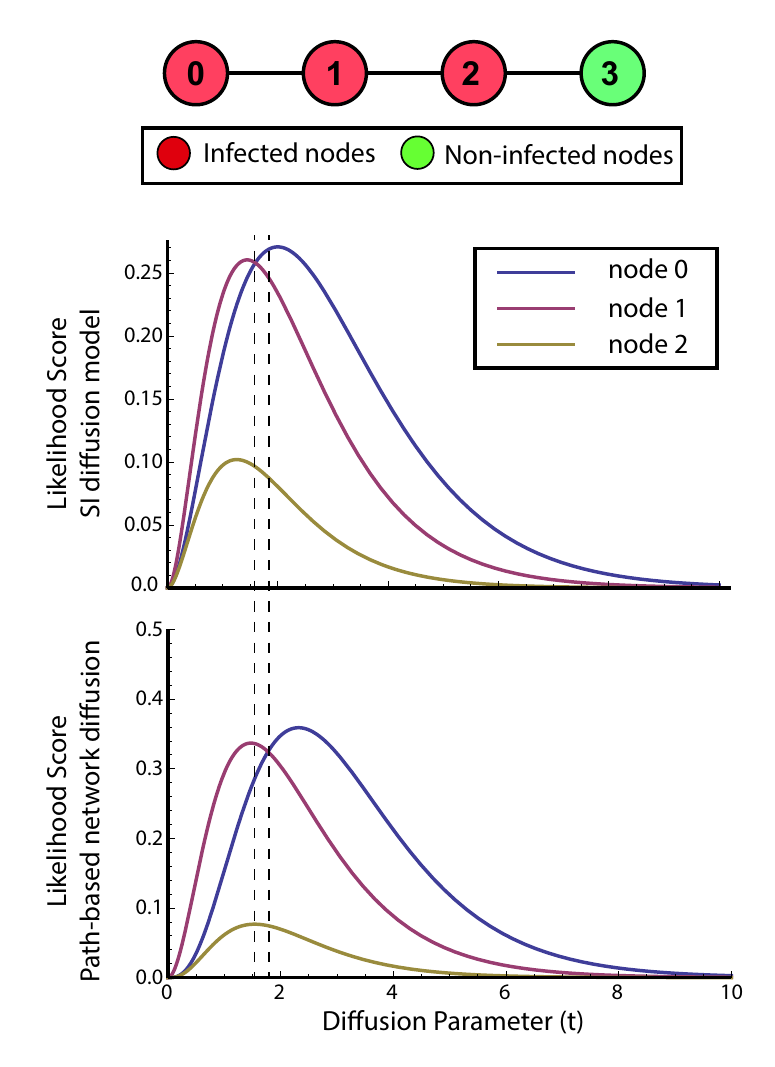}
  \caption{Likelihood scores based on the SI diffusion model and the path-based network diffusion kernel for a depicted line network. Although the underlying diffusion model is SI, NI optimization leads to a tight approximation of the exact solution for a wide range of parameter $t$ (except the range between two vertical dashed lines). Unlike SI, NI computation can be efficiently extended to large complex networks, owing to decoupling of its likelihood terms.}
  \label{fig:line-ex}
\end{figure}

\subsection{Computational Difficulties of the Source Inference Problem}\label{subsec:comp-diff}
In this section, we explain underlying challenges of the source inference problem.

\begin{remark}\label{remark:nonuniqueness}
\textup{Suppose the underlying network $G$ has 4 nodes and 3 undirected edges as depicted in Figure \ref{fig:line-ex}. Suppose the underlying diffusion is according to the SI model of Definition \ref{def:SI-model}. Let the edge holding time variables $\cT_{(i,j)}$ be mutually independent and be distributed according to an exponential distribution with parameter $\lambda$:}
\begin{align}\label{eq:exp-dist}
f_{i,j}(\tau_{i,j})= \lambda e^{-\lambda \tau_{i,j}}, \quad \forall (i,j)\in E.
\end{align}
\textup{Without loss of generality, let $\lambda=1$. Suppose there is a single source in the network (i.e., $m=1$), and we observe the infection pattern at a single snapshot at time $t$. Let the observed infection pattern at time $t$ be $\by(t)=(1,1,1,0)$, implying that nodes $\{0,1,2\}$ are infected at time $t$, while node $3$ is not infected at that time. Our goal is to find the most likely source node, according to the ML Optimization \ref{opt:ml-setup}. Let $\tau_i$ be a realization of the random variable $\cT_i$ (i.e., the time that virus arrives at node $i$). If node $0$ was the source node (i.e., $\tau_0=0$), we would have $\tau_1 \leq \tau_2 \leq t \leq \tau_3$ because of the underlying network structure. Thus,}

\begin{align}\label{eq:line-diffusion}
& Pr\big(\by(t)=(1,1,1,0)|\by(0)=(1,0,0,0)\big)\\
= & \int_{\tau_1=0}^{t}\int_{\tau_2=\tau_1}^{t}\int_{\tau_3=t}^{\infty} e^{-\tau_1} e^{-(\tau_2-\tau_1)}e^{-(\tau_3-\tau_2)} d\tau_1 d\tau_2 d\tau_3\nonumber\\
= & \int_{\tau_1=0}^{t}\int_{\tau_2=\tau_1}^{t}\int_{\tau_3=t}^{\infty} e^{-\tau_3} d\tau_1 d\tau_2 d\tau_3\nonumber\\
= & \frac{1}{2} t^2 e^{-t}.\nonumber
\end{align}
\noindent
\textup{Similarly, we have,}
\begin{align} \label{eq:line-diffusion2}
Pr\big(\by(t)=(1,1,1,0)|\by(0)=(0,1,0,0)\big)=& t(1-e^{-t})e^{-t},\\
Pr\big(\by(t)=(1,1,1,0)|\by(0)=(0,0,1,0)\big)=& e^{-t}-(1+te^{-2t})\nonumber.
\end{align}
\noindent
\textup{These likelihood functions are plotted in \ref{fig:line-ex}. For a given observation time stamp $t$, the ML source estimator selects the node with the maximum likelihood score as the source node, according to Optimization \ref{opt:ml-setup}. Note that an optimal source solution depends on the observation time parameter $t$ (i.e., for $t\lesssim 1.6$, node $1$ and for $t\gtrsim 1.6$, node $0$ are ML optimal source nodes.)}
\end{remark}

\begin{remark}\label{remark:path-dependency}
\textup{Suppose $G$ is a network with $5$ nodes and $5$ edges as shown in Figure \ref{fig:example-graphs}-a. Consider the same diffusion setup as the one of Remark \ref{remark:nonuniqueness}. Let $\by(t)=(1,1,1,1,0)$; i.e., nodes $\{0,1,2,3\}$ are infected at time $t$ while node $4$ is not infected at that time. Similarly to Remark \ref{remark:nonuniqueness}, let $\tau_i$ be a realization of the random variable $\cT_{i}$, a variable representing the time that virus arrives at node $i$. If node $0$ was the source node (i.e., $\tau_0=0$), we would have $\tau_1 \leq \min(\tau_2,\tau_3) \leq \tau_4$,  $\max(\tau_2,\tau_3)\leq t$, and $\tau_4 > t$. Thus,}

\begin{align}\label{eq:path-variable-dependecy}
&Pr\big(\by(t)=(1,1,1,1,0)|\by(0)=(1,0,0,0,0)\big)\\
= & 2 e^{-t}-e^{-2t}\big(1+(1+t)^2\big).\nonumber
\end{align}
\noindent
\textup{In this case, likelihood computation is more complicated than the case of Remark \ref{remark:nonuniqueness}, because both variables $\cT_2$ and $\cT_3$ depend on $\cT_1$, and therefore, consecutive terms do not cancel as in \eqref{eq:line-diffusion}. Moreover, note that there are two paths from node $0$ to node $4$ that overlap at edge $(0,1)$. As we have mentioned earlier, such overlaps are a source of difficulty in the source inference problem, which is illustrated by this simplest example, because the variable $\cT_4$ depends on both variables $\cT_2$ and $\cT_3$ through a $\min(.,.)$ function which makes computation of the likelihood integral complicated.}
\end{remark}

\begin{remark}\label{remark:multi-source}
\textup{Remarks \ref{remark:nonuniqueness} and \ref{remark:path-dependency} explain underlying source inference challenges for a single source node in the network. The case of having multiple source nodes has additional complexity because likelihood scores of Optimization \eqref{opt:ml-setup} should be computed for all possible subsets of infected nodes. For the case of having $m$ sources in the network, there are ${|\VI^{t}|\choose m}$ candidate source sets where for each of them a likelihood score should be computed. If there are a significant number of infected nodes in the network (i.e., $V^{t}=\cO(n)$), there would be $\cO(n^m)$ source candidate sets. This makes the multi-source inference problem computationally expensive for large networks, even for small values of $m$.}
\end{remark}
\noindent
Moreover, in Remarks \ref{remark:nonuniqueness} and \ref{remark:path-dependency}, we assume that the edge holding time distribution is known and follows an exponential distribution with the same parameter for all edges. This is the standard diffusion model used in most epidemic studies \cite{newman-epidemic2002spread}, because the exponential distribution has a single parameter and is memoryless. However, in some practical applications, the edge holding time distribution may be unknown and/or may vary for different edges. We discuss this case in Section \ref{subsec:hetro-NI}.

\begin{figure}[t]
  \centering
      \includegraphics[width=0.3\textwidth]{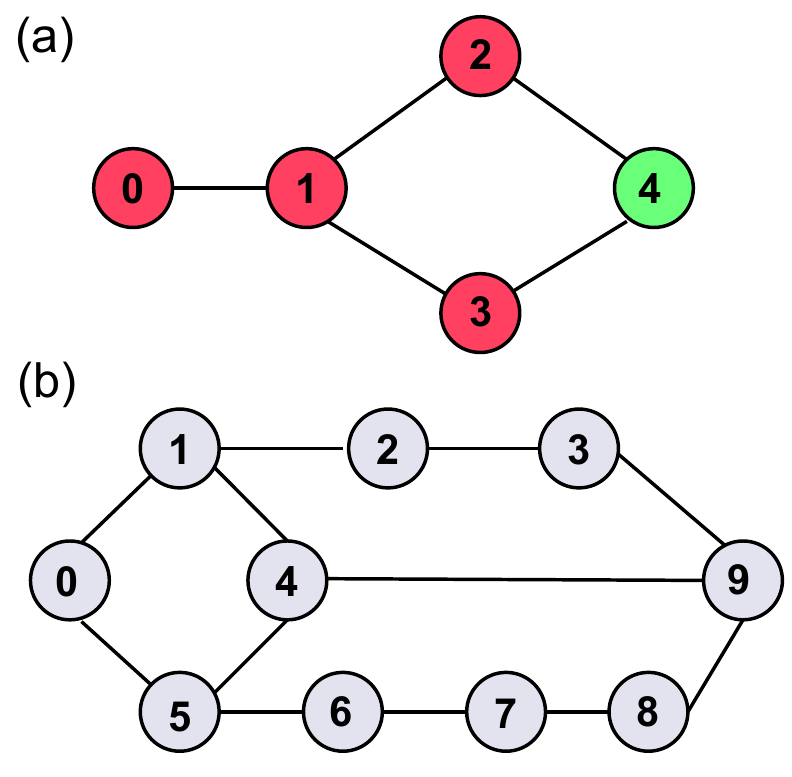}
  \caption{(a) A graph considered in Remark \ref{remark:path-dependency}. (b) An example graph with overlapping shortest paths between nodes $0$ and $9$.}
  \label{fig:example-graphs}
\end{figure}

\subsection{Prior Work}\label{subsec:prior-work}

While our approach considers a general network diffusion setup and its inverse problem, most of the literature considers application to specific problems. The most common ones focus on studying different models of virus propagation in population networks. A standard information diffusion model in this setup is known as the susceptible-infected-recovered (SIR) model \cite{SIR}. At any time, nodes have three types in this model: susceptible nodes which are capable of getting infected, infected nodes that spread virus in the network, and recovered nodes that are cured and can no longer become infected. Under the SIR diffusion model, infection spreads from sources to susceptible nodes probabilistically. References \cite{newman-epidemic2002spread,newman-percolation-epidemics,ganesh2005effect,scale-free-epidemics} discuss the relationship among network structure, infection rate, and the size of the epidemics under this diffusion model. Learning different diffusion parameters of this model have been considered in references \cite{demiris2005bayesian,tutorial-epidemic2002tutorial,okamura2007statistical}. Some other diffusion methods use random walks to model information spread and label propagation in networks \cite{mostafavi2012labeling,Label-propagation,newman2003mixing}. In these models, a random walker goes to a neighbor node with a probability inversely related to node degrees. Therefore, high degree nodes may be less influential in information spread in the network which may be counter-intuitive in some applications.

Although there are several works on understanding mechanisms of diffusion processes in different networks, there is somehow little work on studying the inverse diffusion problem to infer information sources, in great part owing to the presence of path multiplicity in the network \cite{newman2001scientific}, described in Remarks \ref{remark:nonuniqueness}, \ref{remark:path-dependency}, and \ref{remark:multi-source}. Recently, \cite{discrete-diffusion} considers the inverse problem of a diffusion process in a network under a discrete time memoryless diffusion model when time steps are known, while reference \cite{Label-propagation} investigates the problem of identify seed nodes (effectors) of a partially activated network in the steady-state of an Independent-Cascade model. Moreover reference \cite{farajtabar2015back} has studied this problem using incomplete diffusion traces by maximizing the likelihood of the trace under the learned model. The diffusion models and the problem setup considered in these works are different from the continuous-time dynamic diffusion setup considered in this paper. Furthermore, reference \cite{prakash2012spotting} uses the Minimum Description Length principle to identify source nodes in the network. Their method uses a heuristic cost function which combines the model cost and the data cost and lacks theoretical performance characterization.

For the case of having a single source node in the network, some methods infer the source node based on distance centrality \cite{newman2010networks}, or degree centrality \cite{book-graphtheory-west2001introduction} measures of the infected subgraph. These methods are efficient to apply to large networks. However, because they do not assume any particular diffusion model, their performance lacks provable guarantees in general. For tree structures and under a homogeneous SI diffusion model, reference \cite{shah2011rumors} computes a maximum likelihood solution for the source inference problem and provides provable guarantees for its performance. Over tree structures, their solution is in fact equivalent to the distance centrality of the infected subgraph.

The multi-source inference problem has an additional combinatorial complexity compared to the single-source case (see Remark \ref{remark:multi-source}). Reference \cite{lappas2010finding} has studied the complexity of this problem for different types of graphs under an Independent Cascade (IC) model. It shows that for arbitrary graphs this problem is NP-hard to even approximate. However, for some special cases, they introduce a polynomial time algorithm based on dynamic-programming to solve the multi-source problem under the IC model. We study this problem under the SI model for incoherent sources in Section \ref{subsec:NI-multiple-source}.

\section{Main Results}\label{sec:model}
In this section, first, we introduce a path-based network diffusion kernel which is used in proposed source inference methods. Then we present algorithms to infer single and multiple sources in the network. Finally, we present NI algorithms for heterogeneous diffusion and non-parametric cases. For the sake of description, we shall, as before, have a recurrent example of the virus infection spread in the network with the understanding that our framework can be used to solve a more general source inference problem.

\subsection{The path-based Network Diffusion Kernel}\label{subsec:kernel}
Consider the case when there exists a single source node in the network. The multi-source generalization treated in Section \ref{subsec:NI-multiple-source}. Suppose the network structure $G=(V,E)$ is known, and a snapshot $\by(t)$ from some {\it real} diffusion dynamics is observed at time $t$. In general, there may be several diffusion processes that lead to the observed infection snapshot in the network, either exactly or approximately. Suppose $\hby(t')$ is the sample generated at time $t'$ using a certain diffusion model. One way to characterize the error of this diffusion model to explain the observed diffusion sample is to use an asymmetric Hamming premetric function as follows:

\begin{align}\label{eq:diffusion-hamming-apx}
\min_{t'} h_{\alpha}(\by(t),\hby(t')) \triangleq (1-\alpha) \sum_{i:y_i(t)=1} \bbI_{\hy_i=0}+ \alpha \sum_{i:y_i(t)=0} \bbI_{\hy_i(t')=1},
\end{align}
\noindent
where $0\leq \alpha\leq 1$. This error metric assigns weight $\alpha$ to false positive and weight $1-\alpha$ to false negatives error types. To solve the inverse problem, one may select a diffusion process which approximates the observed diffusion pattern closely and also leads to a tractable source inference method. Although the SI diffusion model may be well-suited to model the forward problem of information diffusion in the network, solving the inverse problem (the source inference problem) under this model is challenging in general, in great part owing to the presence of path multiplicity in the network, as we explain in Remarks \ref{remark:nonuniqueness}, \ref{remark:path-dependency} and \ref{remark:multi-source}. Here, we present a path-based network diffusion kernel that is distinct from the standard SI diffusion models, but its order of diffusion approximates well many of them. We will show that this kernel leads to an efficient source inference method with theoretical performance guarantees, under some general conditions.

In our diffusion model, instead of the full network, we consider up to $k$ edge-disjoint shortest paths among pairs of nodes, neglecting other paths in the network. Suppose $\cpij^1$, $\cpij^2$, ... represent different paths between nodes $i$ and $j$ in the network. The length of a path $\cpij^r$ is denoted by $|\cpij^r|$. Let $\Eij^r$ be the set of edges of the path $\cpij^r$. We say two paths are edge-disjoint if the set of their edges do not overlap. Let $\{\cpij^1,\cpij^2,\ldots,\cpij^k\}$ represent $k$ disjoint shortest paths between nodes $i$ and $j$. We choose these paths iteratively so that

\begin{itemize}
\item $|\cpij^1|\leq |\cpij^2|\leq \ldots \leq |\cpij^k|$,
\item paths are disjoint. I.e., for $1<r\leq k$, $\Eij^r \bigcap \big(\bigcup_{a=1}^{r-1}\Eij^a\big)=\emptyset$,
\item $\cpij^r$ is a shortest path between nodes $i$ and $j$ in the network $G'=(V,E-\bigcup_{a=1}^{r-1}\Eij^a)$.
\end{itemize}

In some cases, the shortest path solutions may not be unique. That is, there are at least two shortest paths connecting nodes $i$ to $j$ in the network. If these shortest paths do not overlap, the resulting path length vector $(|\cpij^1|,\ldots,|\cpij^k|)$ is the same irrespective of the selection order. Thus, the tie breaking can be done randomly. However, in the case of having overlapping shortest paths, one way to break the tie among these paths is to choose the one which leads to a shorter path in the next step. For example, consider the network depicted in Figure \ref{fig:example-graphs}-b. There are two paths of length $3$ between nodes $0$ and $9$. Choosing the path $0-5-4-9$ leads to the next independent path $0-1-2-3-9$ with length $4$, while choosing the path $0-1-4-9$ leads to the next path $0-5-6-7-8-9$ of length $5$. Thus, the algorithm chooses the path $0-5-4-9$. If next paths have the same length, tie would be broken considering more future steps. In practice, this case has negligible effect in the performance of the source inference method. Methods based on message-passing or dynamic programming can be used to select optimal $k$ shortest paths in the network as well \cite{kpash1,kpash2}. In this paper, we break ties randomly among paths with the same length.

Recall that $\cT_{\cpij^r}$ represents the virus traveling time over the path $\cpij^r$ whose cumulative density function is denoted by $\Fijr(.)$ according to Equation \eqref{eq:cdf-general-path}.

\begin{definition}[Path-based network diffusion kernel]\label{def:path-diffusion-kernel}
Let $p_{i,j}(t)$ be the probability of node $j$ being infected at time $t$ if node $i$ is the source node. Thus

\begin{align}\label{eq:path-diffusion-kernel}
p_{i,j}(t)= & Pr\big[y_j(t)=1|y_i(0)=1\big]\\
\triangleq & 1- \prod_{r=1}^{k} \big(1-\Fijr(t)\big),\nonumber
\end{align}
\noindent
where $k$ is the number of independent shortest paths between nodes $i$ and $j$. $P(t)=[p_{i,j}(t)]$ is called a path-based network diffusion kernel.
\end{definition}

In the path-based network diffusion kernel node $j$ gets infected at time $t$ if the infection reaches to it over at least one of the $k$ independent shortest paths connecting that node to the source node. The path-based network diffusion kernel provides a non-dynamic diffusion basis for the network and is based on two important assumptions: that edge holding time variables $\cT_{(i,j)}$ are mutually independent; and the paths are disjoint. A path-based network diffusion kernel with $k=1$ only considers the shortest paths in the network and has the least computational complexity among other path-based network diffusion kernels. Considering more paths among nodes in the network (i.e., $k>1$) provides a better characterization of network diffusion processes with the cost of increased kernel computational complexity (Proposition \ref{prop:complexity-diffusion-kernel}). Computation of the path-based network diffusion kernel compared to the
SI diffusion model is provably efficient even for large and complex networks.

In the following, we highlight properties and relaxations of the path-based network diffusion kernel:

The path-based network diffusion kernel provides a non-dynamic diffusion model, where nodes become infected independently based on their distances (path lengths) to source nodes. Unlike the dynamic SI model, in the path network diffusion model, it is possible (though unlikely) to have $y_i(t)=1$ while $y_j(t)=0$, for all neighbors of node $i$ (i.e., $j\in\cN(i)$). The key idea is that to infer the source node in the network, full characterization of diffusion dynamics, in many cases, may not be necessary as long as the diffusion model approximates the observed samples closely (e.g., according to an error metric of \eqref{eq:diffusion-hamming-apx}). For instance, consider the setup of Remark \ref{remark:nonuniqueness} where the underlying diffusion model is according to a SI model. In that example, we compute source likelihood scores in \eqref{eq:line-diffusion} and \eqref{eq:line-diffusion2} by integrating likelihood conditional density functions. The likelihood computation under this model becomes challenging for complex networks. However, according to the path-based network diffusion model of Definition \ref{def:path-diffusion-kernel}, these likelihood scores are decoupled to separate terms and can be computed efficiently as follows:

\begin{align}\label{eq:line-diffusion3}
Pr\big(\by(t)=(1,1,1,0)|\by(0)=(1,0,0,0)\big)=&F(1,t)F(2,t)\Fbar(3,t),\\
Pr\big(\by(t)=(1,1,1,0)|\by(0)=(0,1,0,0)\big)=&F(1,t)^2\Fbar(2,t),\nonumber\\
Pr\big(\by(t)=(1,1,1,0)|\by(0)=(0,0,1,0)\big)=&F(1,t)F(2,t)\Fbar(1,t),\nonumber
\end{align}
\noindent
where $F(l,t)$ is the Erlang cumulative distribution function over a path of length $l$, that we shall show in \eqref{eq:erlang-simple-notation}. Figure \ref{fig:line-ex} shows likelihood scores of infected nodes computed according to \eqref{eq:line-diffusion3}. This example illustrates that for a wide range of parameter $t$, both models lead to the same optimal solution. Moreover, unlike the SI model, likelihood functions can be computed efficiently using the path-based network diffusion kernel, even for large complex networks.

The path-based diffusion kernel considers only the top $k$ shortest paths among nodes, neglecting other paths in the networks. The effect of long paths is dominated by the one of short ones leading to a tight approximation. Suppose $\cpij^1$ and $\cpij^2$ represents two paths between nodes $i$ and $j$ where $|\cpij^1| \ll |\cpij^2|$ (i.e., the path $\cpij^2$ is much longer than the path $\cpij^1$). Thus, for a wide range of parameter $t$, we have $\Fijo(t) \gg \Fijt(t)$, and therefore

\begin{align}\label{eq:apx-k-path-diffusion}
(1-\Fijo(t))(1-\Fijt(t))\approx 1-\Fijo(t).
\end{align}

Note that for very small or large $t$ values (i.e., $t\to 0$ or $t\to \infty$), both $\Fijo(.)$ and $\Fijt(.)$ go to 0 and 1, respectively, and thus the approximation \eqref{eq:apx-k-path-diffusion} remains tight. For an example network depicted in Figure \ref{fig:kernel-apx}-a, we illustrate the tightness of the first order approximation (i.e., $k=1$) for different lengths of the path $\cpij^2$. In general for large $k$ values the gap between the approximate and the exact kernels becomes small with the cost of increased kernel computational complexity (see Proposition \ref{prop:complexity-diffusion-kernel}). The same approximation holds for overlapping paths with different lengths as it is illustrated in Figures \ref{fig:kernel-apx}-b.

\begin{figure}[t]
  \centering
      \includegraphics[width=0.3\textwidth]{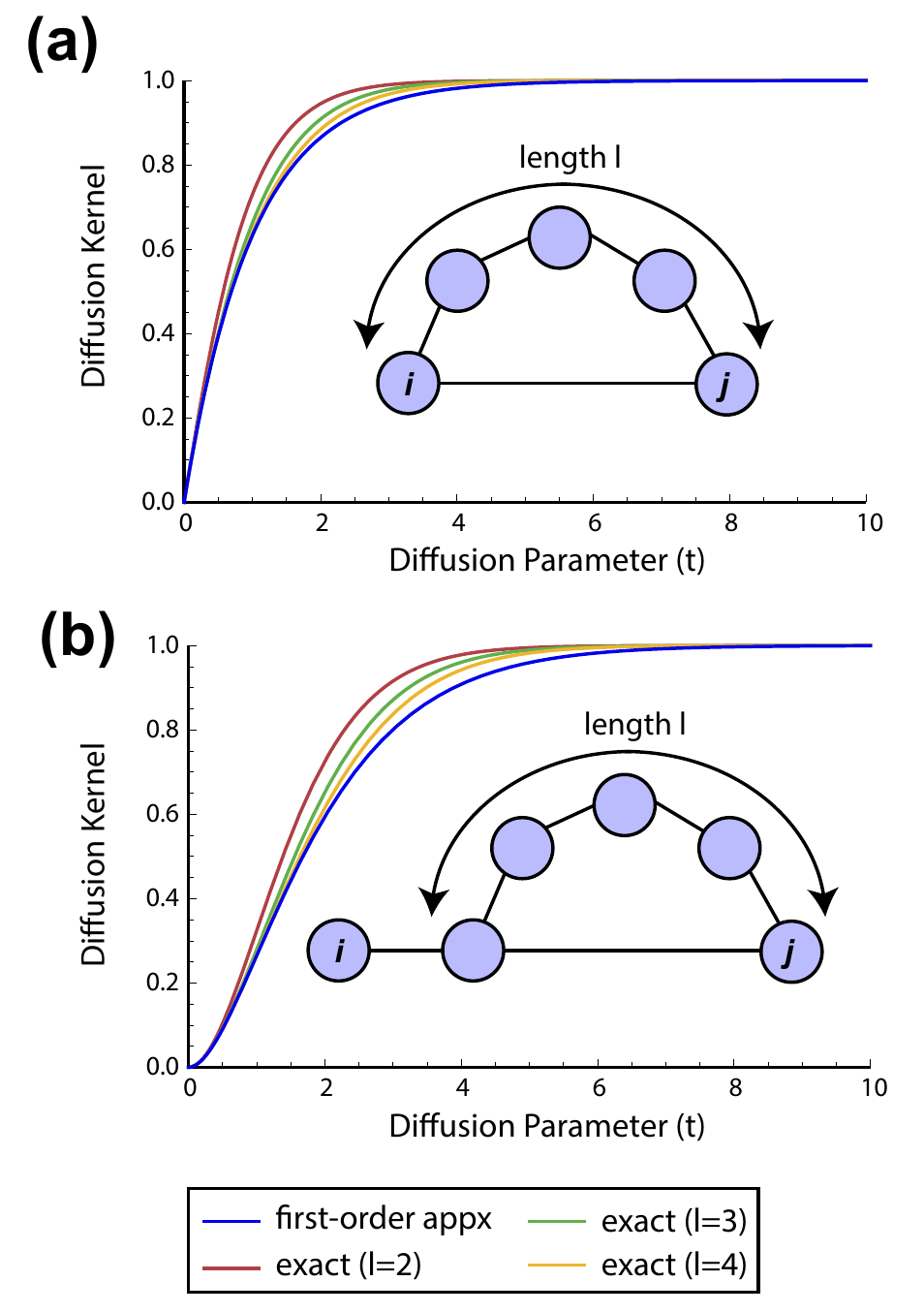}
  \caption{Tightness of the first order approximation of the path-based network diffusion kernel over example networks.}
  \label{fig:kernel-apx}
\end{figure}

Finally, path-based network diffusion kernel only considers independent shortest paths among nodes and therefore ignores the effects of non-disjoint paths in the network. This is a critical relaxation because, as we explain in Remark \ref{remark:path-dependency}, overlapping paths and dependent variables make the likelihood computation and therefore source inference challenging. In general, if there are many overlapping shortest paths among nodes in the network, this approximation might not be tight. However, in network structures whose paths do not overlap significantly (for example tree structures), this approximation is tight. In Appendix \ref{SIsec:erdos}, we show that in a common model for sparse random networks \cite{erdHos1961strength}, shortest paths among nodes that are the dominating factors in the path-based network diffusion kernel computation, are extremely unlikely to overlap with each other, leading to a tight kernel approximation.

One main advantage of using the path-based network diffusion kernel compared to other diffusion models such as the SI diffusion model is its efficient computation even for large and complex networks:

\begin{proposition}[Computational complexity of path-based network diffusion kernel]\label{prop:complexity-diffusion-kernel}
Let $G=(V,E)$ be a directed network with $n$ nodes and $|E|$ edges. Then, computation of the $k$-path network diffusion kernel of Definition \ref{def:path-diffusion-kernel} has a worst case computational complexity $\cO(k|E|n+kn^2\log(n))$.
\end{proposition}

\begin{remark}
To solve the source inference problem, one only needs to compute rows of the path-based network diffusion kernel which correspond to infected nodes ($i\in V^{t}$). Thus, time complexity of kernel computation can be reduced to $\cO(k|E||\VI^{t}|+k|\VI^{t}|n\log(n))$, where $|\VI^{t}|$ is the number of observed infected nodes in the network at time $t$.
\end{remark}
\noindent

Computation of the path-based network diffusion kernel depends on edge holding time distributions. If the virus traveling time variables $\cT_{(i,j)}$ are identically distributed for all edges in the network, the underlying diffusion process is called {\it homogeneous}. On the other hand, if holding time distributions differ among edges in the network, the resulting diffusion process is {\it heterogeneous}. In this section, we consider a homogeneous diffusion setup, where the holding time distribution is an exponential distribution with parameter $\lambda$ for all edges. Without loss of generality, we assume that $\lambda=1$.
It is important to note that the proposed framework is not limited to use of either the exponential or a single delay distribution.
Other edge holding time distributions can be used in the framework as well. For example, the case of heterogeneous diffusion is considered in Section \ref{subsec:hetro-NI}.

Under the setup considered in this section, the virus traveling time over each path in the network has an Erlang distribution,
because it is the sum of independent exponential variables. Thus, we have that
\begin{align}\label{eq:erlang-kernel}
\Fijr(t)=Pr[\cT_{\cpij^r}\leq t]=\frac{\gamma(|\cpij^r|,\lambda t)}{(|\cpij^r|-1)!},
\end{align}
where $\gamma(\cdot)$ is the lower incomplete gamma function \cite{bertsekas2002introduction}. $|\cpij^r|$ (the path length connecting node $i$ to $j$) is also called the Erlang's shape parameter. Because $\Fijr(t)$ is only a function of the path length and parameter $t$, to simplify notation, we define
\begin{align}\label{eq:erlang-simple-notation}
F(l,t)\triangleq\Fijr(t),
\end{align}
where $l=|\cpij^r|$. The $k$-path network diffusion kernel of Definition \ref{def:path-diffusion-kernel} using the Erlang distribution is called a path-based Erlang network diffusion kernel. If only one shortest path among nodes is considered (i.e., $k=1$), the diffusion kernel of Definition \ref{def:path-diffusion-kernel} is called the shortest path network diffusion kernel.

\begin{definition}[Shortest path Erlang diffusion kernel]\label{def:Erlang-kernel-shortest-path}
The shortest path Erlang network diffusion kernel is defined as follows:

\begin{align}\label{eq:shortestpath-erlang-diffusion-kernel}
p_{i,j}(t)= Pr\big[y_j(t)=1|y_i(0)=1\big]= F(d_{i,j},t),
\end{align}
\noindent
where $d_{i,j}$ is the length of the shortest path connecting node $i$ to node $j$, and $F(d_{i,j},t)$ represents the Erlang cumulative distribution function of \eqref{eq:erlang-simple-notation}.
\end{definition}
\noindent
The shortest path Erlang diffusion kernel can be viewed as the first order approximation of the underlying diffusion process. It has the least computational complexity among other path-based network diffusion kernels which makes it suitable to be used over large and complex networks. Moreover, this kernel has a single parameter $t$ which can be learned reliably using the observed samples (see Section \ref{subsec:parameters}).
\subsection{Global Source Inference}\label{subsec:NI-global}
In this section, we describe a source inference method called Network Infusion (NI) which aims to solve the inverse diffusion problem over a given network using observed infection patterns. The method described in this section finds a single node as the source of the {\it global} information spread in the network. In Section \ref{subsec:NI-multiple-source}, we consider the case when more than one source node exists, where each source causes a {\it local} infection propagation in the network. In this section, we also assume that the infection pattern is observed at a single snapshot at time $t$ (i.e., $\by(t)$ is given). The case of having multiple snapshots is considered in Appendix \ref{SIsec:multi-snap-shot}.

\noindent
Recall that $\VI^{t}$ is the set of observed infected nodes at time $t$, and $P(t)=[p_{i,j}(t)]$ represents the path-based network diffusion kernel according to Definition \ref{def:path-diffusion-kernel}. Under the path-based network diffusion kernel, the joint diffusion
probability distribution can be decoupled into individual marginal
distributions, which leads to a tractable ML Optimization, even for
complex networks:

\begin{alg}\label{alg:NI-ML-single}
A maximum-likelihood NI algorithm (NI-ML) infers the source node by solving the following
optimization:
\begin{align}\label{opt:NI-ML}
&\arg\max_{i\in\VI^{t}} \cL(i,t)\\
=&\arg\max_{i\in\VI^{t}} \sum_{j\in\VI^{t}} \log\big(p_{i,j}(t)\big)+\sum_{j\notin\VI^{t}} \log\big(1-p_{i,j}(t)\big).\nonumber
\end{align}
\end{alg}

In Optimization \eqref{opt:NI-ML}, $\cL(i,t)$ is called the NI log-likelihood function (score) of node $i$ at time $t$. In this optimization, the parameter $t$ is the time at which the observation is made and is assumed to be known. For cases that this parameter is unknown, we introduce techniques to learn a within-diffusion-model time parameter with provable performance guarantees in various setups (Section \ref{subsec:parameters}).

To have a well-defined NI log-likelihood objective function, $p_{i,j}(t)$ should be non-zero for infected nodes $i$ and $j$ (i.e.,  $p_{i,j}(t)\neq 0$ when $i,j\in\VI^{t}$). If infected nodes form a strongly connected sub-graph over the network, this condition is always satisfied. In practice, if $p_{i,j}(t)= 0$ for some $i,j\in\VI^{t}$ (i.e., $i$ and $j$ are disconnected in the graph), we assume that $p_{i,j}(t)=\eps$ (e.g., $\eps=10^{-6}$). Note that for $j\notin\VI^{t}$, for any value of $t>0$, $p_{i,j}(t)<1$. Therefore, the second term in the summation of Optimization \eqref{opt:NI-ML} is always well-defined.

\noindent
NI Algorithm \ref{alg:NI-ML-single} aims to infer the source node by maximizing $\cL(i,t)$. An alternative approach is to infer the source node by minimizing the expected prediction error of the observed infection pattern. We describe this approach in the following:
\noindent
Let $h_{\alpha}(\by,\bx)$ be a weighted Hamming premetric between two binary sequences $\bx$ and $\by$ defined as follows:

\begin{align}
h_{\alpha}(\by,\bx)\triangleq (1-\alpha) \sum_{i:y_i=1} \bbI_{x_i=0}+ \alpha \sum_{i:y_i=0} \bbI_{x_i=1},
\end{align}
\noindent
where $0\leq \alpha\leq 1$.
\noindent
If $\alpha=1/2$, $h_{\alpha}(.,.)$ is a metric distance. If $\alpha\neq 1/2$, $h_{\alpha}(.,.)$ is a premetric (not a metric) because it does not satisfy the symmetric property of distance metrics (i.e., there exist $\bx$ and $\by$ such that $h_{\alpha}(\bx,\by)\neq h_{\alpha}(\by,\bx)$), and it does not satisfy the triangle inequality as well (i.e., there exist $\bx$, $\by$ and $\bz$ such that $h_{\alpha}(\bx,\by)> h_{\alpha}(\by,\bz)+h_{\alpha}(\bz,\bx)$).

\begin{remark}\label{remark:alpha}
\textup{$h_{\alpha}(\by,\bx)$ generalizes Hamming distance between binary sequences $\by$ and $\bx$ using different weights for different error types. Suppose $\bx$ is a prediction of the sequence $\by$. There are two types of possible errors: (1) if $y_i=1$ and the prediction is zero (i.e., $x_i=0$, false negative error), (2) if $y_i=0$ and the prediction is one (i.e., $x_i=1$, false positive error). $h_{\alpha}(\by,\bx)$ combines these errors by assigning weight $1-\alpha$ to false negatives (i.e., missing ones), and weight $\alpha$ to false positives (i.e., missing zeros). Having different weights for different error types can be useful, specially if the sequence $\by$ is sparse. Suppose $\by$ has $\kappa$ ones (positives), and $n-\kappa$ zeros (negatives). Therefore, there are $\kappa$ possible type 1 errors and $n-\kappa$ type 2 errors in prediction. In this case, to have balance between the number of true negative and false positive errors, one can choose $\alpha=\kappa/n$ in calculation of $h_{\alpha}(\by,\bx)$.}
\end{remark}
\noindent
Now we introduce a NI algorithm which infers the source node by minimizing the prediction error.

\begin{alg}\label{alg:NI-ME-single}
A minimum  error NI algorithm (NI-ME) infers the source node by solving the following optimization:
\begin{align}\label{opt:NI-ME}
& \arg\min_{i\in\VI^{t}} \cH_{\alpha}(i,t)=\arg\min_{i\in\VI^{t}} \bbE[h_{\alpha}(\by(t),\bx_i(t))]\\
=& \arg\min_{i\in\VI^{t}} (1-\alpha)\sum_{j\in\VI^{t}}\big(1-p_{i,j}(t)\big)+\alpha\sum_{j\notin\VI^{t}} p_{i,j}(t),\nonumber
\end{align}
\noindent
where $\bx_i(t)$ is a binary prediction vector of node $i$ at time $t$ with probability distribution $P_i(t)$, and $\cH_{\alpha}(i,t)$ is the expected prediction error of node $i$ at time $t$.
\end{alg}
\noindent
Similarly to Maximum Likelihood NI (NI-ML) Algorithm, we assume that the parameter $t$ (the time at which observation is made) is known. We discuss the case when this parameter is unknown in Section \ref{subsec:parameters}.

\begin{remark}\label{remark:alpha2}
\textup{According to Remark \ref{remark:alpha}, to have balance between false positive and false negative error types, one can use $\alpha=|\VI^{t}|/n$ where $|\VI^{t}|$ is the number of infected nodes (positives) at time $t$. However, in general, this parameter can be tuned in different applications using standard machine learning techniques such as cross validations \cite{cross-validation}.}
\end{remark}
\noindent
The proposed NI methods based on maximum likelihood (NI-ML, Algorithm \ref{alg:NI-ML-single}) and minimum error (NI-ME, Algorithm \ref{alg:NI-ME-single}) are efficient to solve even for large complex networks:

\begin{proposition}\label{prop:complexity-NI}
Suppose the underlying network $G=(V,E)$ has $n$ nodes and $|E|$ edges. Let $\VI^{t}$ represent the set of infected nodes at time $t$. Then, a worst case computational complexity of NI Algorithms \ref{alg:NI-ML-single} and \ref{alg:NI-ME-single} is $\cO\big(|\VI^{t}|(k|E|+kn\log(n))\big)$.
\end{proposition}

\noindent
In the rest of this section, we analyze the performance of NI Algorithms \ref{alg:NI-ML-single} and \ref{alg:NI-ME-single} under a standard SI diffusion model of Definition \ref{def:SI-model}.

\begin{theorem}\label{thm:optimality-NI-ML}
Let $G=(V,E)$ be an undirected tree with countably infinite nodes. Suppose node $s$ is the source node, $t$ is the infection observation time, and the underlying diffusion process is according to the SI model of Definition \ref{def:SI-model}. Then, we have,
\begin{align}\label{eq:NI-ML-inequality}
\bbE[\cL(s,t)]\geq \bbE[\cL(i,t')], \quad \forall i,\forall t',
\end{align}
\noindent
where $\bbE[\cL(i,t')]$ is the expected NI log-likelihood score of node $i$ with parameter $t'$.
\end{theorem}
In the setup of Theorem \ref{thm:optimality-NI-ML}, similarly to the setup of reference \cite{shah2011rumors}, we assume that the set of vertices is countably infinite to avoid boundary effects. Theorem \ref{thm:optimality-NI-ML} provides a mean-field (expected) optimality for Algorithm \ref{alg:NI-ML-single}. In words, it considers the case when we have sufficient samples from independent infection spreads in the network starting from the same source node.

Note that the function $\cL(i,t)$ provides a statistics of node $i$ being the source node. Theorem \ref{thm:optimality-NI-ML} considers the expectation of this function under the SI diffusion model. Since both expectation and summation are linear operations, the graph structure is tree, and in $\cL(i,t)$ nodes have been decoupled, computing this expectation under the SI model and the path-based network diffusion kernel leads to the same argument (for more details, see the proof of the Theorem). However by considering the expectation under the path-based network diffusion kernel, the argument of Theorem \ref{thm:optimality-NI-ML} holds for a general graph structure.

In \eqref{eq:NI-ML-inequality}, $i$ can be equal to $s$ (the source node), and/or $t'$ can be equal to $t$ as well. If $i$ is equal to $s$, we have the following:

\begin{corollary}\label{prop:robustness-NI-ML}
Under the conditions of Theorem \ref{thm:optimality-NI-ML}, we have,

\begin{align}
\bbE[\cL(s,t)]\geq \bbE[\cL(s,t')], \quad \forall t',
\end{align}
\noindent
where the equality holds iff $t=t'$.
\end{corollary}

\begin{remark}\label{remark:NI-ML-t-parameter}
\textup{In this remark, we highlight the difference between parameters $t$ and $t'$ in Theorem \ref{thm:optimality-NI-ML}.
The parameter $t$ is the time at which we observe the infection pattern in the network. If this parameter is known, it can be used to compute likelihood scores according to Optimization \eqref{opt:NI-ML}. However, this parameter may be unknown and one may use an estimate of this parameter in Optimization \eqref{opt:NI-ML} (i.e., using $t'$ instead of $t$). Theorem \ref{thm:optimality-NI-ML} indicates that even if different parameters $t'\neq t$ are used to compute source likelihood scores for different nodes, the likelihood score obtained by the source node $s$ and the true parameter $t$ is optimal in expectation. This theorem and corresponding Proposition \ref{prop:robustness-NI-ML} provide a theoretical basis to estimate the underlying true parameter $t$ by maximizing the likelihood score for each node over different values of $t'$ (for more details, see Section \ref{subsec:parameters}).}
\end{remark}

In the following, we present the mean-field optimality of minimum error NI algorithm (NI-ME) over regular tree structures:

\begin{theorem}\label{thm:optimality-NI-ME}
Let $G=(V,E)$ be a regular undirected tree with countably infinite nodes. Suppose node $s$ is the source node, $t$ is the observation time, and the underlying diffusion process is according to the SI model of Definition \ref{def:SI-model}. Then, for any value of $0<\alpha<1$ and $t'>0$, we have,
\begin{align}\label{eq:NI-ME-inequality}
\bbE[H_{\alpha}(s,t')]<\bbE[H_{\alpha}(i,t')], \quad \forall i\neq s, \forall t'>0,
\end{align}
\noindent
where $\bbE[H_{\alpha}(i,t')]$ is the expected prediction error of node $i$ using parameter $t'$. Equality \eqref{eq:NI-ME-inequality} holds iff $s=i$.
\end{theorem}

The mean field optimality of NI-ME algorithm holds for all values of $0<\alpha<1$ under the setup of Theorem \ref{thm:optimality-NI-ME}. In practice and under more general conditions, we find that $\alpha$ selection according to Remarks \ref{remark:alpha} and \ref{remark:alpha2} leads to a robust performance, owing to the balance between true negative and false positive errors (Sections \ref{sec:sim} and \ref{sec:digg}).

\begin{remark}\label{remark:NI-ME-t-parameter}
\textup{The NI-ML mean-field optimality of Theorem \ref{thm:optimality-NI-ML} holds even if different $t'$ values are used for different nodes. However, the mean-field optimality of the NI-ME method of Theorem \ref{thm:optimality-NI-ME} holds if the same $t'$ parameter is used for all nodes. Interestingly, even if the parameter used in the NI-ME algorithm is difference than the true observation time parameter (i.e., $t'\neq t$), the optimality argument of Theorem \ref{thm:optimality-NI-ME} holds which indicates the robustness of the method with respect to this parameter. Moreover, the NI-ME optimality of inequality \eqref{eq:NI-ME-inequality} is strict, while the one of NI-ML method according to the inequality \eqref{eq:NI-ML-inequality} may have multiple optimal solutions.}
\end{remark}

\subsection{Multi-source NI}\label{subsec:NI-multiple-source}

In this section, we consider the multi-source inference problem, where there exists $m$ sources in the network. We consider this problem when sources are sufficiently distant from each other and only a single snapshot, at time $t$, is available (i.e., $\by(t)$ is given). For simplicity of the analysis, we consider $k=1$ in the path-based network diffusion kernel (i.e., only shortest paths are considered).

Let $G=(V,E)$ be the underlying network where $d_{i,j}$ represents the length of the shortest path between node $i$ and node $j$. Define $D(i,R)\triangleq \{j\in V| d_{i,j}< R\}$ as a disk with radius $R$ centered at node $i$, which we refer to as the $R$-neighborhood of node $i$ in the network. Similarly, the union of disks with radius $R$ centered at nodes of the set $V_1\subset V$ is defined as $D(V_1,R)\triangleq \{j\in V| \exists i\in V_1, d_{i,j}< R\}$. We define the following distances in the network:

\begin{align}\label{eq:d0-d1}
d_{0}&\triangleq \arg\max_{d} F(d,t)>\frac{1}{2},\\
d_{1}^{\eps}&\triangleq \arg\min_{d} F(d,t)<\frac{\eps}{nm},\nonumber
\end{align}
\noindent
where $F(d,t)$ is defined according to \eqref{eq:erlang-simple-notation}.

\begin{definition}[$\eps$-Coherent Sources]\label{def:eps-coherent}
Let $G=(V,E)$ be a binary network. Sources $\cS=\{s_1,s_2,\ldots,s_m\}$ are $\eps$-coherent if,
\begin{align}
d(s_a,s_b)>2(d_0+d_1^{\eps}), \quad \forall 1\leq a,b\leq m, a\neq b,
\end{align}
where $d_0$ and $d_1$ are defined according to \eqref{eq:d0-d1}.
\end{definition}
\noindent
Intuitively, sources are incoherent if they are sufficiently distant from each other in the network so that their infection effects at time $t$ do not overlap in the network (for instance, viruses released from them, with high probability, have not visited the same nodes.). This assumption is a critical condition to solve the multi-source NI problem efficiently.

\begin{definition}[Multi-Source Network Diffusion Kernel]\label{def:kernel-multi-source}
Suppose $G=(V,E)$ is a possibly directed binary graph and there exist $m$ source nodes $\cS=\{s_1,\ldots,s_m\}$ in the network that are $\eps$-coherent. We say a node $j\notin\cS$ gets infected at time $t$ if it gets a virus from at least one of the sources. Thus, we have,

\begin{align}\label{eq:kernel-multiple-sources}
Pr[y_j(t)=1]\triangleq 1-\prod_{s\in\cS}\barp_{s,j}(t),
\end{align}
where $\barp_{s,j}(t)=1-p_{s,j}(t)$.
\end{definition}
\noindent
Using multi-source network diffusion kernel of Definition \ref{def:kernel-multi-source}, the log-likelihood function $\cL(\cS,t)$ and the Hamming error function $\cH_{\alpha}(\cS,t)$ are defined as follows:

\begin{align}\label{eq:multiple-likeli-error-functions}
\cL(\cS,t)\triangleq&\sum_{j\in\VI^{t}} \log(1-\prod_{s\in\cS}\barp_{s,j}(t))+\sum_{j\notin\VI^{t}} \log(\prod_{s\in\cS}\barp_{s,j}(t)),\\
\cH_{\alpha}(\cS,t)\triangleq&(1-\alpha)\sum_{j\in\VI^{t}} \prod_{s\in\cS}\barp_{s,j}(t)+\alpha\sum_{j\notin\VI^{t}} \big(1-\prod_{s\in\cS}\barp_{s,j}(t)\big).\nonumber
\end{align}
\noindent
Similarly to Algorithms \ref{alg:NI-ML-single} and \ref{alg:NI-ME-single}, NI aims to find a set of $m$ sources which maximizes the log-likelihood score, or minimizes the weighted Hamming error. However, unlike the single source case, these optimizations are computationally costly because all ${|\VI^{t}|\choose m}$ possible source combinations should be evaluated. If the number of infected nodes is significant ($|\VI^{t}|=\cO(n)$), even for small constant number of sources, one needs to compute the likelihood or error scores for approximately $\cO(n^m)$ possible source subsets, which may be computationally overly challenging for large networks.
\noindent
One way to solve this combinatorial optimization is to take an iterative approach, where, at each step, one source node is inferred. However, at each step, using single source NI methods may not lead to an appropriate approximation because single source NI methods aim to find the source node which explains the entire infection pattern in the network, while in the multi-source case, the entire infection pattern are caused by multiple sources. To avoid this problem, at each step, we use a localized version of NI methods developed in Algorithms \ref{alg:NI-ML-single} and \ref{alg:NI-ME-single}, where sources explain the infection pattern only around their neighborhood in the network.

\begin{definition}
The localized likelihood function of node $i$ in its $d_0$ neighborhood is defined as,
\begin{align}\label{eq:localized-likelihood-function}
\cL_{d_0}(i,t)\triangleq \sum_{\substack{
j\in\VI^{t}\\
j\in D(i,d_0)}} \log\big(p_{i,j}(t)\big)+\sum_{\substack{
j\notin\VI^{t}\\
j\in D(i,d_0)}} \log\big(1-p_{i,j}(t)\big),
\end{align}
where only nodes in the $d_0$ neighborhood of node $i$ is considered in likelihood computation.
\end{definition}
\noindent
A similar argument can be expressed for the localized Hamming prediction error. For large $d_0$ values, the localized likelihood score is similar to the global likelihood score of \eqref{opt:NI-ML}. Using localized likelihood function is important in the multi-source NI problem because source candidates cannot explain the infection pattern caused by other sources. In the following, we propose an efficient localized NI method to solve the multi-source inference problem by maximizing localized likelihood scores of source candidates using a greedy approach. A similar algorithm can be designed for the localized minimum error NI.

\begin{alg}[Multi-source NI-ML Algorithm]\label{alg:multi-source--NI}
Suppose $\cS_r$ is the set of inferred sources at iteration $r$. The localized NI-ML algorithm has the following steps:
\begin{itemize}
\item {\it Step 0:} $\cS_0=\emptyset$.
\item {\it Step r+1:}
\begin{itemize}
\item {\it Likelihood computation:} compute $s_{r+1}$ using the following optimization,
\begin{align}\label{opt:ml-multi-source-greedy}
s_{r+1}&=\arg\max_{i\in\VI^{t}-D(\cS_k,d_1^{\eps})}  \cL_{d_0}(i,t).
\end{align}
\item {\it Update the source set}: add $s_{r+1}$ to the list of inferred sources,
\begin{align}
\cS_{r+1}=\cS_r\cup s_{r+1},\nonumber
\end{align}
\end{itemize}
\item {\it Termination:} stop if $r=m$.
\end{itemize}
\end{alg}
\noindent
In the following, we show that if sources are sufficiently incoherent (i.e., sufficiently distant from each other in the network), the solution of localized NI Algorithm \ref{alg:multi-source--NI} approximates the exact solution closely.

\begin{theorem}\label{thm:optimality-multi-source}
Let $G=(V,E)$ be a regular undirected tree with countably infinite nodes. Suppose sources are $\eps$-coherent according to Definition \ref{def:eps-coherent}, and the underlying diffusion process is according to the SI model of Definition \ref{def:SI-model}. Suppose $\cS_r$ is the set of sources inferred by localized NI Algorithm \ref{alg:multi-source--NI} till iteration $r$. If $\cS_k\subset\cS$, then with probability at least $1-\eps$, there exists a source node that has not been inferred yet whose localized likelihood score is optimal in expectation:
\begin{align}
&\exists s\in \cS-\cS_r,\quad,\bbE[\cL_{d_0}(s,t)]\geq \bbE[\cL_{d_0}(i,t)]\\
&\forall i\in \VI^{t}-D(\cS_k,d_1^{\eps}).\nonumber
\end{align}
\end{theorem}
\begin{proposition}\label{prop:complexity-greedy-NI}
A worst case computational complexity of localized NI Algorithm \ref{alg:multi-source--NI} is $\cO\big(|\VI^{t}|(k|E|+kn\log(n)+mn)\big)$.
\end{proposition}
\subsection{NI for Heterogeneous Network Diffusion}\label{subsec:hetro-NI}
In previous sections, we have assumed that the infection spread in the network is homogeneous; i.e., virus traveling time variables $\cT_{(i,j)}$ are i.i.d. for all edges in the network. This can be an appropriate model for the binary (unweighted) graphs. However, if edges have weights, the infection spread in the network may be heterogeneous; i.e., the infection spread is faster over strong connections compared to the one of weak edges.
\noindent
Suppose $G=(V,E,W)$ represents a weighted graph, where $w(i,j)>0$ if $(i,j)\in E$, and $w(i,j)=0$ otherwise. One way to model a heterogeneous diffusion in the network is to assume that edge holding time variables $\cT_{(i,j)}$ are distributed independently according to an exponential distribution with parameter $\lambda_{i,j}=w(i,j)$. According to this model, the average holding time of edge $(i,j)$ is $1/w_{i,j}$, indicating the fast spread of infection over strong connections in the network.

\begin{figure}[t]
\begin{center}
\includegraphics[width=0.4\textwidth]{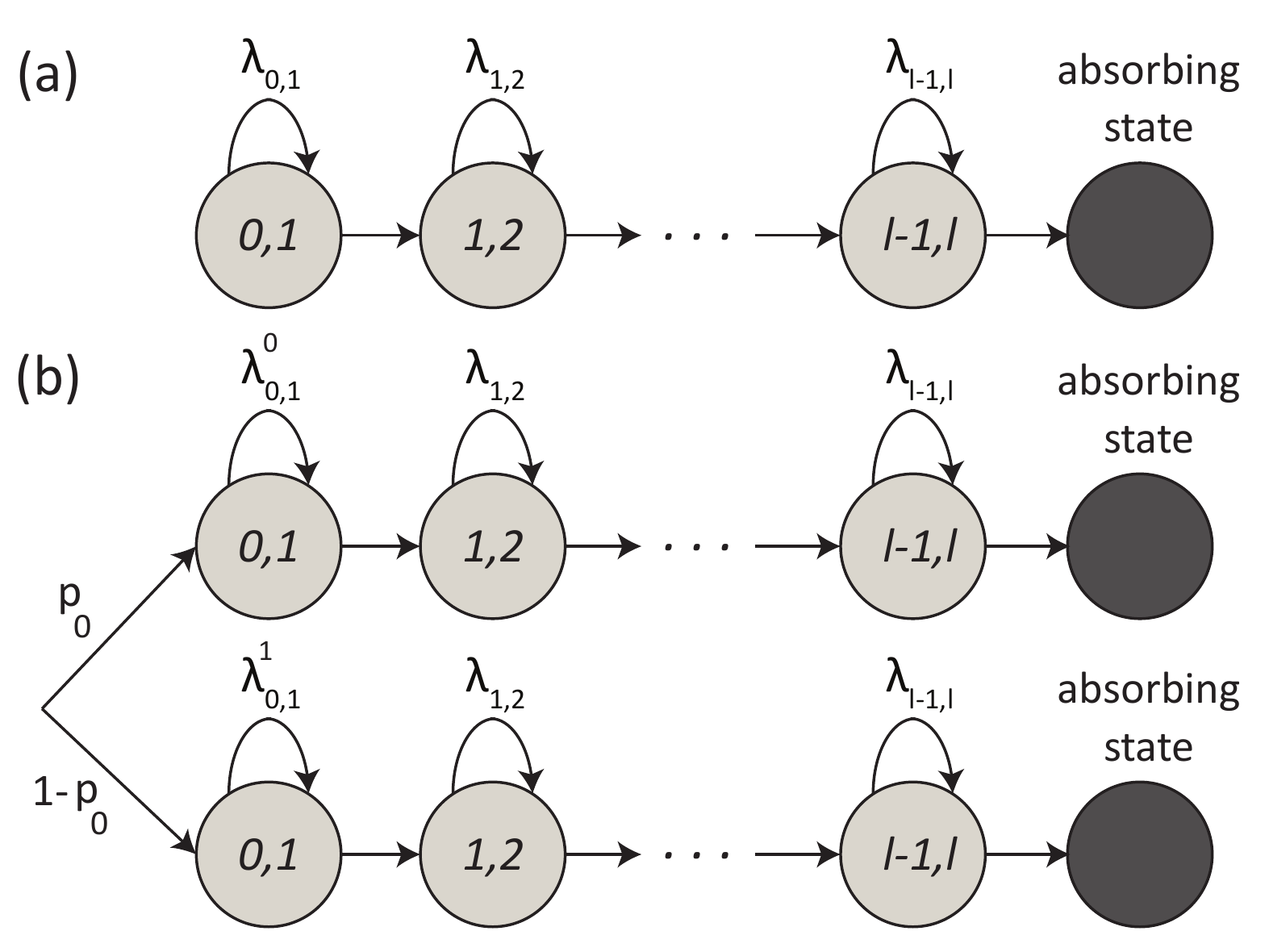}
\end{center}
\caption{(a) A Markov chain of a hypo-exponential distribution. (b) A Markov chain of a mixed hypo-exponential distribution.}
\label{fig:markov}
\end{figure}
\noindent
Recall that $\cT_{\cpij^r}$ represents the virus traveling time variable from node $i$ to node $j$ over the path $\cpij^r$. To simplify notations and highlight the main idea, consider the path $\mathcal{P}_{0\to l}^r=\{0\to 1 \to 2 \ldots \to l\}$. The virus traveling time from node $0$ to node $l$ over this path ($\cT_{\mathcal{P}_{0\to l}^r}$) is a {\it hypoexponential} variable whose distribution is a special case of the {\it phase-type distribution}. For this path, we consider a Markov chain with $l+1$ states, where the first $l$ states are transient, and the state $l+1$ is an absorbing state. Each transient state of this Markov chain corresponds to an edge $(i,j)$ over this path whose holding time is characterized by an exponential distribution with rate $\lambda_{i,j}=w(i,j)$ (Figure \ref{fig:markov}-a). In this setup, the virus traveling time from node $0$ to node $l$ over the path $\mathcal{P}_{0\to l}^r$ is equal to the time from the start of the process until reaching to the absorbing state of the corresponding Markov chain. The distribution of this absorbing time can be characterized as a special case of the phase-type distribution. A subgenerator matrix of the Markov chain of Figure \ref{fig:markov}-a is defined as follows:

\begin{align}
\left[\begin{matrix}-\lambda_{0,1}&\lambda_{0,1}&0&\dots&0&0\\
                    0&-\lambda_{1,2}&\lambda_{1,2}&\ddots&0&0\\
                    \vdots&\ddots&\ddots&\ddots&\ddots&\vdots\\
                    0&0&\ddots&-\lambda_{l-3,l-2}&\lambda_{l-3,l-2}&0\\
                    0&0&\dots&0&-\lambda_{l-2,l-1}&\lambda_{l-2,l-1}\\
                    0&0&\dots&0&0&-\lambda_{l-1,l}
\end{matrix}\right].\
\end{align}
\noindent
For simplicity, denote the above matrix by $\Theta\equiv\Theta(\lambda_{0,1},\dots,\lambda_{l-1,l})$. Define $\boldsymbol{\alpha}=(1,0,\dots,0)$ as the probability of starting in each of the $l$ states. Then, the Markov chain absorbtion time is distributed according to $PH(\boldsymbol{\alpha},\Theta)$, where $PH(.,.)$ represents a phase-type distribution. In this special case, this distribution is also called a hypoexponential distribution. A similar subgenerator matrix $\Theta$ can be defined for a general path $\cpij^r$ connecting nodes $i$ to $j$. Thus, we have

\begin{figure*}[t]
  \centering
      \includegraphics[width=1.0\textwidth]{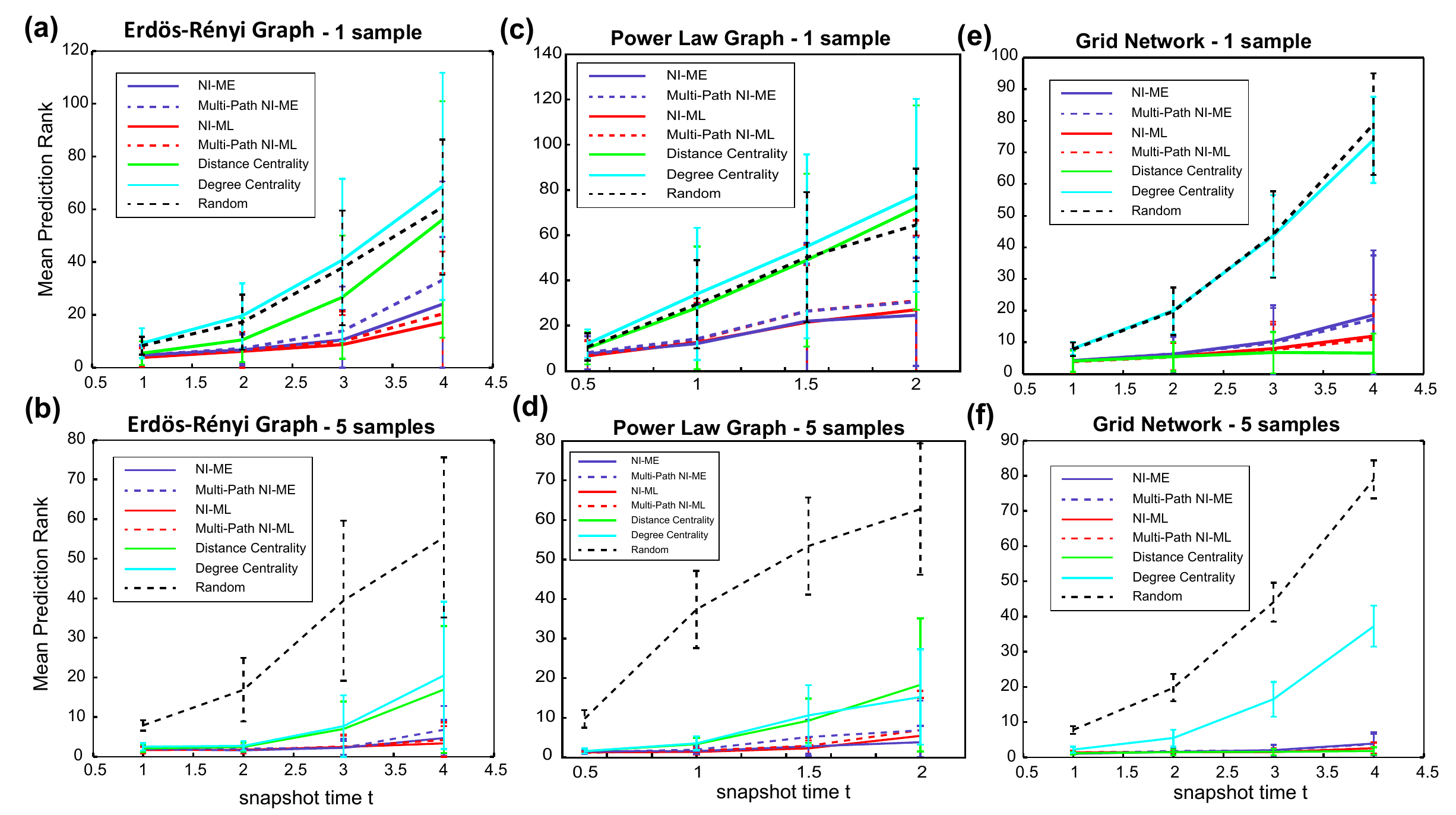}
  \caption{Performance of source inference methods over (a,b) Erd\"os-R\'enyi, (c,d) power law and (e,f) grid networks. For evaluation, we compute the rank of true sources averaged over different runs of simulations. If this score is close to one, it means that the true source is among top predictions of the method. Error bars represent $\pm$ standard deviations of empirical results for each case. Experiments have been repeated 100 times so that error margins are small. For more details, see Section \ref{sec:sim} and Appendix \ref{SIsec:results-sim}.}
  \label{fig:main-sim}
\end{figure*}

\begin{align}\label{eq:hypo-kernel}
\Fijr(t)=Pr[\cT_{\cpij^r}\leq t]=1-\boldsymbol{\alpha}e^{t\Theta}\boldsymbol{1},
\end{align}
\noindent
where $\boldsymbol{1}$ is a column vector of ones of the size $|\cpij^r|$, and $e^{X}$ is the matrix exponential of $X$. For an unweighted graph where all edges have the same rate $\lambda$, \eqref{eq:hypo-kernel} is simplified to \eqref{eq:erlang-kernel}. For the weighted graph $G=(V,E,W)$, we compute $k$ shortest paths among pairs of nodes over the graph $G'=(V,E,W')$, where $w'(i,j)=1/w(i,j)$ if $(i,j)\in E$, otherwise $w'(i,j)=\infty$. Then, the path network diffusion kernel for a weighted graph $G=(V,E,W)$ can be defined according to Definition \eqref{def:path-diffusion-kernel}. Using this kernel, NI algorithms introduced in Sections \ref{subsec:NI-global} and \ref{subsec:NI-multiple-source} can then be used to infer the source node under the heterogeneous diffusion in the network.
\noindent
Note that this framework can be extended to a more complex diffusion setup as well. We provide an example of such diffusion setup in the following:

\begin{example}
\textup{Consider the the path $\mathcal{P}_{0\to l}^r=\{0\to 1\to 2 \to \ldots \to l\}$. Suppose the edge $(0,1)$ spreads the infection with rates $\lambda_{0,1}^0$ and $\lambda_{0,1}^1$ with probabilities $p_0$ and $1-p_0$, respectively. Suppose other edges $(i,j)$ of this path spread the infection with rate $\lambda_{i,j}$. Figure \ref{fig:markov}-b illustrates the corresponding Markov chain for this path. The subgenerator matrices of this Markov chain can be characterized as follows:}

\begin{align}
\Theta_{i}=\left[\begin{matrix}-\lambda_{0,1}^i&\lambda_{0,1}^i&0&\dots&0&0\\
                    0&-\lambda_{1,2}&\lambda_{1,2}&\ddots&0&0\\
                    \vdots&\ddots&\ddots&\ddots&\ddots&\vdots\\
                    0&0&\ddots&-\lambda_{l-3,l-2}&\lambda_{l-3,l-2}&0\\
                    0&0&\dots&0&-\lambda_{l-2,l-1}&\lambda_{l-2,l-1}\\
                    0&0&\dots&0&0&-\lambda_{l-1,l}
\end{matrix}\right],\
\end{align}
\textup{for $i=0,1$. Then, for this path, we have}

\begin{align}\label{eq:hypo-kernel}
F_{\mathcal{P}_{0\to l}^r}(t)=Pr[\cT_{\mathcal{P}_{0\to l}^r}\leq t]=1-\boldsymbol{\alpha}\big(p_0e^{t\Theta_0}+(1-p_0)e^{t\Theta_1}\big)\boldsymbol{1}.
\end{align}
\noindent
\textup{To compute the path-based network diffusion kernel, we compute shortest paths among pairs of nodes over the graph $G'=(V,E,W')$, where $w'(0,1)=p_0/\lambda_{0,1}^0+(1-p_0)/\lambda_{0,1}^1$, $w'(i,j)=1/w(i,j)$ if $(i,j)\in E$ and $(i,j)\neq (0,1)$, and $w'(i,j)=\infty$ otherwise. This example illustrates that NI framework can be used even under a complex heterogeneous diffusion setup.}
\end{example}

\subsection{Non-parametric NI - Unknown Snapshot Time}\label{subsec:parameters}
In some real-world applications, only the network structure $G$ and infection patterns $\{\by(t_1),\ldots,\by(t_z)\}$ are known and therefore to use NI algorithms, we need to learn the parameters such as observation times $\{t_1,\ldots,t_z\}$, the number of sources $m$ in the network, and the number of disjoint shortest paths considered in the diffusion kernel $k$. In the following, we introduce efficient techniques to learn these parameters:

{\it Observation time parameters:} In the maximum likelihood NI Algorithm \ref{alg:NI-ML-single}, according to Remark \ref{remark:NI-ML-t-parameter}, the true parameter $t$ is the one that maximizes the expected source likelihood score according to Theorem \ref{thm:optimality-NI-ML}. Thus, in the case of unknown parameter $t$, we solve the following optimization:

\begin{align}\label{opt:NI-ML-unknown-t}
(s,t)=\arg\max_{i,t'} \cL(i,t'),
\end{align}
where $\cL(i,t')$ is the log-likelihood function of node $i$ using parameter $t'$. One way to solve Optimization \eqref{opt:NI-ML-unknown-t} approximately is to quantize the range of parameter $t$ (i.e., $t\in (0,t_{max})$) to $b$ bins and evaluate the objective function in each case. Because we assume that the $\lambda$ parameter of the edge holding time distribution is equal to one, one appropriate choice for $t_{max}$ is the diameter of the infected subgraph, defined as the longest shortest path among pairs of infected nodes. The number of quantization levels $b$ determines the resolution of the inferred parameter $t$ and therefore the tightness of the approximation. If $t_{max}$ is large and the true $t$ parameter is small, to have a tight approximation, the number of quantization levels $b$ should be large which may be computationally costly. In this case, one approach to estimate parameter $t$ is to use the first moment approximation of the Erlang network diffusion kernel over source neighbors. Suppose $\mu$ is the fraction of the infected neighbors of source $s$. Since infection probabilities of source neighbors approximately come from an exponential distribution, for a given parameter $t$, $\mu\approx 1-e^{-t}$. Therefore,
    \begin{align}\label{eq:estimate-t-small}
    t\approx -\ ln(1-\mu).
    \end{align}
In the minimum error NI Algorithm \ref{alg:NI-ME-single}, according to Remark \ref{remark:NI-ME-t-parameter}, the prediction error of all infected nodes should be computed using the same parameter $t$. In the setup of Theorem \ref{thm:optimality-NI-ME}, any value of parameter $t$ leads to an optimal solution in expectation. In general, we suggest the following approach to choose this parameter: First, for each node, we minimize the prediction error for different values of the parameter $t$ as follows,

\begin{align}\label{opt:NI-ME-unknown-t}
t_i^*=\arg\max_{t'} \cH_{\alpha}(i,t').
\end{align}

This Optimization can be solved approximately similarly to the case of maximum likelihood NI Optimization \eqref{opt:NI-ML-unknown-t}. For the small $t$ values, we use \eqref{eq:estimate-t-small} to obtain $t_i^*$. Then, to obtain a fixed $t$ parameter for all nodes, we take the median of $t_i^*$ parameters of nodes with the minimum prediction error (in our experiments, we consider top 10 predictions with the smallest error to estimate the parameter $t$). In the cases of multi-source and multi-snapshot NI, one can use similar approaches to estimate time stamp parameters.

{\it The number of sources:} In NI algorithms presented in Section \ref{subsec:NI-multiple-source}, we assume that the number of sources in the network (i.e., the parameter $m$) is known. In the case of unknown parameter $m$, if sources are sufficiently incoherent according to Definition \ref{def:eps-coherent}, one can estimate $m$ as follows: because sources are incoherent, their caused infected nodes do not overlap with each other with high probability. Thus, the number of connected components of the infected sub-graph (or the number of infected clusters in the network) can provide a reliable estimate of the number of sources in this case.

{\it Regularization parameter of NI-ME:} The minimum error NI Algorithm \ref{alg:NI-ME-single} has a regularization parameter $\alpha$ which balances between false positive and false negative error types. In the setup of Theorem \ref{thm:optimality-NI-ME}, any value of $0<\alpha<1$ leads to an optimal expected weighted Hamming error solution of \eqref{eq:NI-ME-inequality}. However, in general, we choose this parameter according to Remarks \ref{remark:alpha} and \ref{remark:alpha2} to have a balance between the number of false negative and false positive errors.

{\it Number of disjoint shortest paths considered in kernel computation:} Our proposed path-based network diffusion kernel considers up to $k$ independent shortest paths among nodes where a node $j$ gets infected at time $t$ if the infection reaches to it over at least one of the $k$ disjoint shortest paths connecting that node to the source node. A path-based network diffusion kernel with $k=1$ has the least computational complexity among other path-based network diffusion kernels. Considering more paths among nodes in the network (i.e., $k>1$) provides a better characterization of network diffusion processes with the cost of increased kernel computational complexity (Proposition
\ref{prop:complexity-diffusion-kernel}). For example, over an asymmetric grid network of Appendix \ref{SIsec:results-sim}, considering increased numbers of paths among nodes to form the path network diffusion kernel improves the performance of NI methods significantly as higher order paths partially capture the asymmetric diffusion spread in the network.

\begin{figure}[t]
  \centering
      \includegraphics[width=0.4\textwidth]{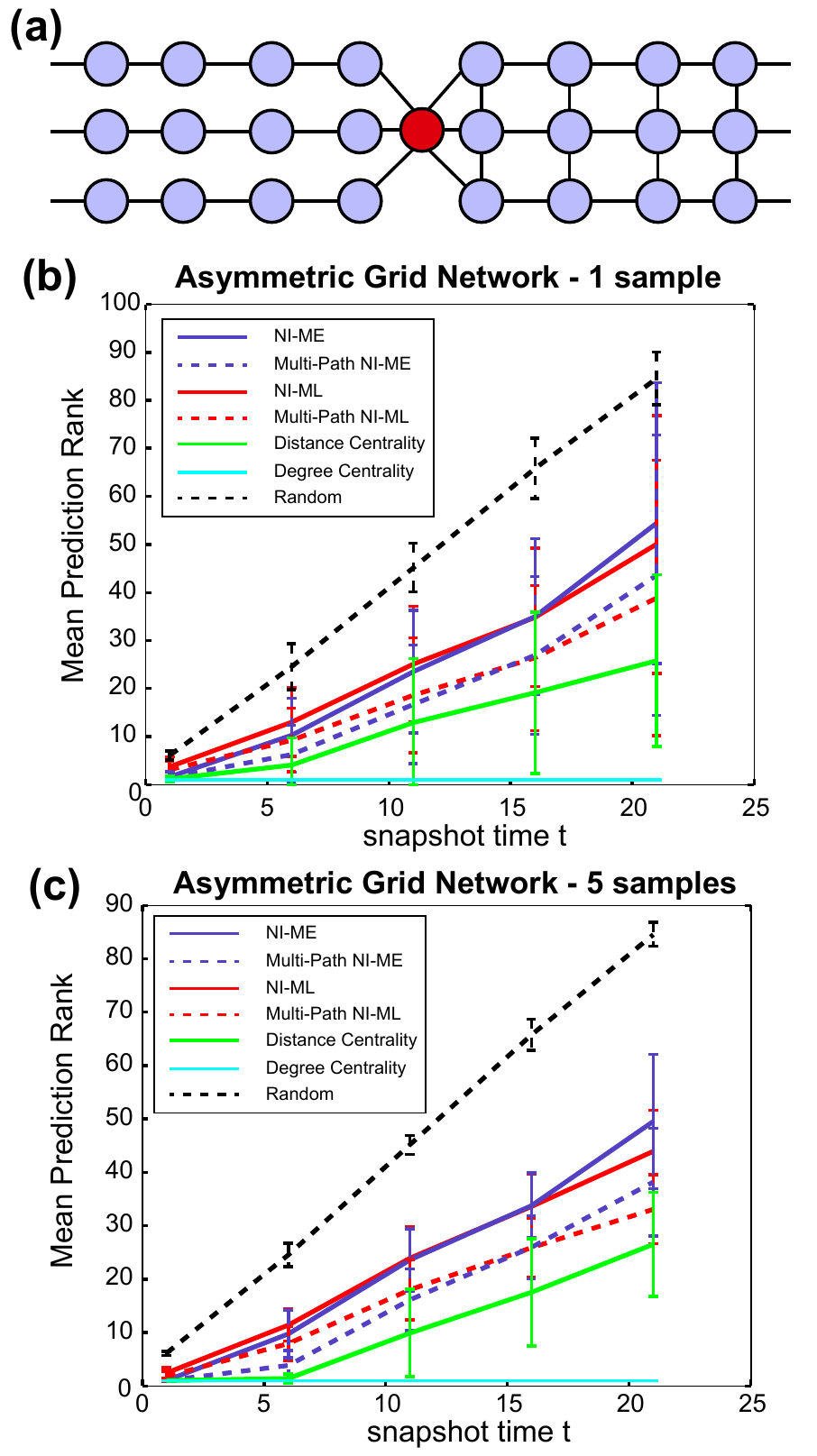}
  \caption{Performance of source inference methods over asymmetric grid networks with $250$ nodes with (b) one and (c) five independent samples. Node with red color is the source node.}
  \label{fig:asym-grid}
\end{figure}
\section{Performance Evaluation on Synthetic and Real Data}
\subsection{NI Over Synthetic Networks}\label{sec:sim}
In this section, we assess the performance of NI and other source inference
algorithms over different synthetic network structures. To
generate simulated diffusion patterns, we use the SI kernel to allow
a fair performance comparison with existing methods. The inputs to the algorithms are the network structure $G$, and the
observed infection pattern at some {\it unknown} time $t$ (i.e.,
$\by(t)$); we estimate a within-diffusion-model parameter $t$ using the observed infection
pattern and network structure according to techniques described in
Section \ref{subsec:parameters}. For evaluation, we sort infected nodes as
source candidates according to the score obtained by different
methods. We then compute the rank of true sources averaged over different
runs of simulations. If this score is close to one, it means that
the true source is among top predictions of the method (For more details, see Appendix \ref{SIsec:results-sim}).
\begin{figure*}[t]
  \centering
      \includegraphics[width=0.8\textwidth]{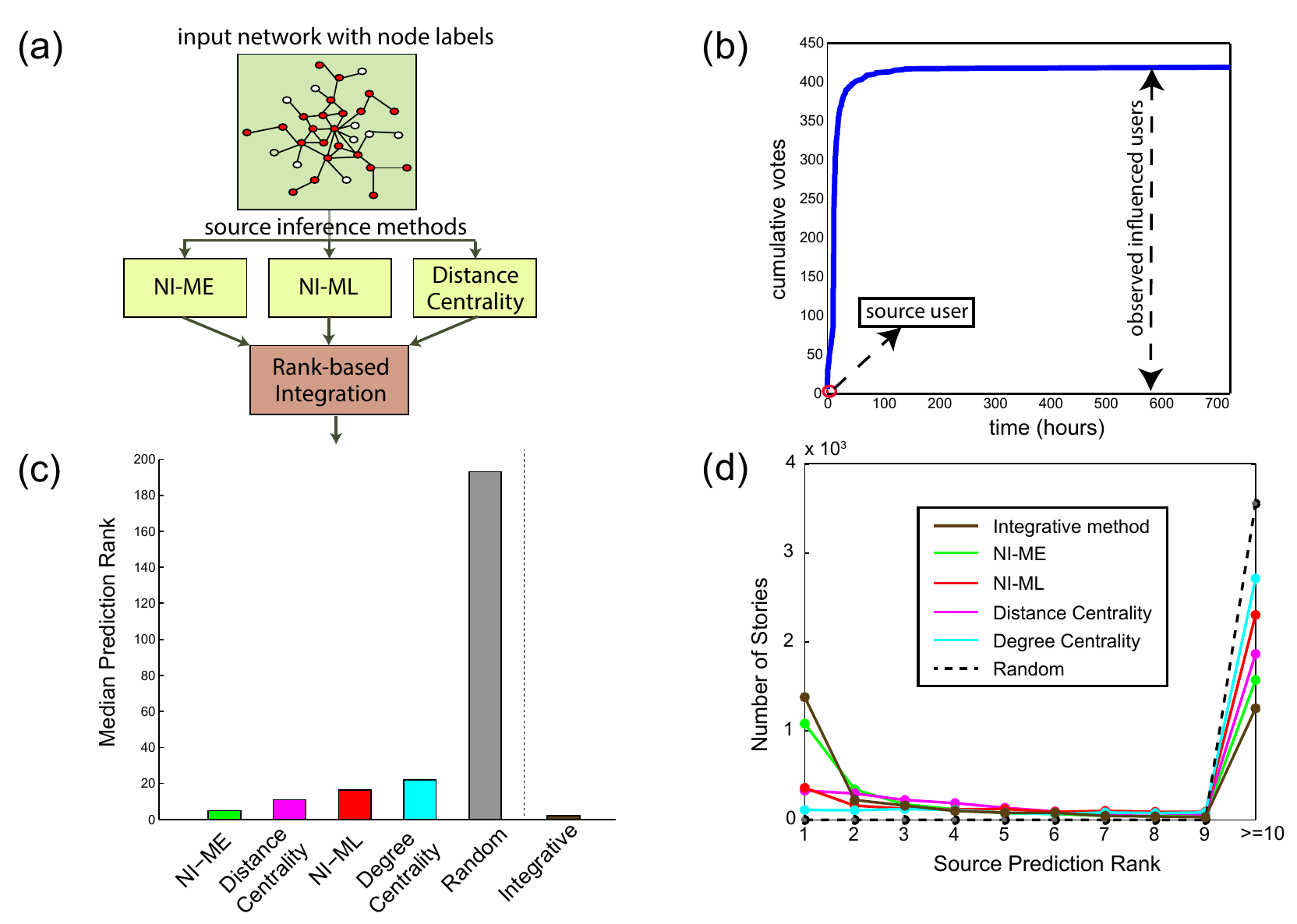}
  \caption{(a) An integrative source inference framework which combines prediction ranks of NI-ME, NI-ML and distance centrality methods. (b) The cumulative number of votes for a particular story over time. (c) Median prediction ranks of true sources for over 3,500 Digg stories inferred by individual and integrative source inference methods. (d) Histogram of prediction ranks of sources inferred by different source inference methods.}
  \label{fig:main-digg}
\end{figure*}
Figure \ref{fig:main-sim}-(a-d) compares the performance of different source inference
methods over both Erd\"os-R\'enyi and power law networks, and in different
ranges of the parameter $t$. Panels (a) and (c) illustrate the performance when only one sample from the infection pattern at time $t$ is available, while Panels (b) and (d) demonstrate the performance when five independent samples from the infection pattern at time $t$ are given, illustrating the mean-field optimality of NI methods according to Theorems \ref{thm:optimality-NI-ML} and \ref{thm:optimality-NI-ME}. In all cases, the NI-ML and NI-ME algorithms significantly
outperform other methods. Because the underlying networks are sparse,
the performance of single-path and multi-path NI methods are close
to each other in both network models. Notably, unlike NI methods,
the performance of other source inference methods such as distance
centrality and degree centrality does not tend to converge to the
optimal value even for higher sample sizes (Appendix \ref{SIsec:results-sim}).

Similar results hold for NI algorithms over grid structures (Figure \ref{fig:main-sim}-e,f). In this case, the performance of the distance centrality method for higher sample sizes converges to the optimal value as well. Since grid structures are symmetric, even though paths among nodes overlap significantly, considering shortest paths approximates the true diffusion dynamics closely (Appendix \ref{SIsec:results-sim}). In Figure \ref{fig:asym-grid}, we also illustrate an adversarial example of an asymmetric grid structure where the performance of source inference methods do not converge to the optimal value, owing to the asymmetric diffusion spread in the network. In this case, considering increased numbers of paths among nodes to form the path network diffusion kernel improves the performance of NI methods significantly, as higher order paths partially capture the asymmetric diffusion spread in the network (Appendix \ref{SIsec:results-sim}).

\subsection{Inference of News Sources in Digg Social News Networks}\label{sec:digg}
We evaluate the performance of NI and other source inference techniques in identifying news sources over the social news aggregator Digg. We also propose an integrative source inference framework which combines network infusion methods (NI-ME and NI-ML) with the distance centrality method. We exclude the degree centrality method from our integration owing to its low performance over synthetic and real networks. In our integrative framework, we compute the average rank of each node to be the source node according to different source inference methods (Figure \ref{fig:main-digg}-a and Appendix \ref{SIsec:digg}). The proposed integrative method removes method-specific biases of individual source inference techniques and leads to robust performance, particularly when the underlying dynamics are unknown.

Reference \cite{digg} has
collected voting activities and friendship connections of Digg's
users over a period of a month in 2009 for $3,553$ promoted news stories.
We use this data to form a friendship network of active Digg users
with $24,219$ nodes and more than $350 K$ connections. We consider
users as active if they have voted for at least 10 stories in this
time period. Figure 1 demonstrates a small part of the Digg
friendship network. Figure Figure \ref{fig:main-digg}-b demonstrates the cumulative number of votes for a particular
story at different times. If friends of a user $A$ vote for a
specific story, it is more likely that user $A$ will also vote for that
story and that is how information propagates over the Digg friendship
network. NI aims to invert this information propagation process
to infer the news source by observing the voting pattern in a single
snapshot in the steady state when no further news propagation occurs. Here, we only consider the shortest
path (i.e. $k=1$) among nodes to compute the path-based network
diffusion kernel used in NI Algorithms \ref{alg:NI-ML-single} and
\ref{alg:NI-ME-single}. This application provides an interesting real data framework to assess
the performance and robustness of different source inference methods
because the underlying diffusion processes are based on real dynamics over the Digg friendship
network.

Figure Figure \ref{fig:main-digg}-c demonstrates the median prediction ranks of true sources inferred by individual and integrative source inference methods. NI-ME outperforms other individual source inference methods methods with a medium prediction rank of 5, while the integrative source inference leads to the best performance with a median prediction rank of 3. In fact, NI-ME  and integrative source inference algorithms infer news sources optimally (in their top prediction) for approximately $30\%$ and $38\%$ of stories, respectively (Figure Figure \ref{fig:main-digg}-d). This illustrates the robustness of these methods in cases where the underlying dynamics are unknown.
\section{Discussion}
The key idea of our proposed source inference method is designing diffusion processes that closely approximate the observed diffusion pattern, while leading to tractable source inference methods for large complex networks. Our proposed path-based network diffusion kernel well-approximates a standard SI diffusion model particularly over sparse networks. The key intuition is that to infer the source node in the network, full characterization of diffusion dynamics, in many cases, may not be necessary. One advantage of our proposed path-based network diffusion kernel is its tunability, that it can consider different number of shortest paths in kernel computation. This resembles, vaguely, a Taylor-expansion of network topology to form a diffusion kernel with different orders of expansion. One can extend this key idea to design other network diffusion kernels to approximate other general diffusion models such as SIR (susceptible-infected-recovered)\cite{SIR}, or to design network-specific diffusion kernels considering different topological properties such as their symmetry, degree distribution, etc.

The NI framework infers source nodes using the network topology and snapshots from the infection spread. On the other hand,
\cite{rodriguez2014uncovering,sefer2015convex} have considered the problem of network inference given propagation pathways over the network. If the network topology is partially unknown or has some errors with either false positive or false negative edges, one can devise a {\it joint network inference-network infusion} framework where in one direction, the pattern of infection spread is used to learn the network topology by denoising false positive and false negative edges, while in the other direction, the topology of the network is used to infer pathways of infection spread over the network.

\begin{appendices}

\section{Tightness of the Path-based Network Diffusion Kernel for Sparse Erd\"os-R\'enyi Graphs}\label{SIsec:erdos}
The path-based network diffusion kernel of Definition \ref{def:Erlang-kernel-shortest-path} does not consider effect of overlapping paths among nodes. This can cause the SI approximation to be loose particularly if there are several overlapping shortest paths among nodes in the network since shortest paths are dominant factors in kernel computation. In the following proposition we show that in a sufficiently sparse Erd\"os-R\'enyi graph, the probability of having overlapping shortest paths among nodes is small.
\begin{proposition}\label{prop:erdos-overlap}
Let $G=(V,E)$ be an undirected Erd\"os-R\'enyi graph with $n$ nodes where $Pr[(i,j)\in E]=p$. Consider two nodes $i$ and $j$ where $d(i,j)\leq l_0$. If $p<\frac{c}{n}$ where $c=1/n^{2l_0}$, the probability of having overlapping shortest paths between nodes $i$ and $j$ goes to zero asymptotically.
\end{proposition}
\begin{proof}
First, we compute the probability that a node $v\in V$ belongs to a cycle of length at most $l$. Such a cycle is determined by the $l-1$ other vertices. By choosing them in order, there are less than $n^{l-1}$ choices for those other vertices, while the cycle appears with probability $p^l$ in the graph. Thus, the probability that $v$ is involved in a cycle of length $l$ is at most $n^{l-1}p^l \leq c^l/n$. To have an overlapping shortest path between nodes $i$ and $j$, at least one of the nodes over that path should belong to a cycle of length at most $2l_0$. This happens with probability less than $l_0 c^{2l_0}/n$, which goes to zero asymptotically if $c<1/n^{2l_0}$.
\end{proof}
Figure \ref{fig:erdos-path-overlap} illustrates Proposition \ref{prop:erdos-overlap} for Erd\"os-R\'enyi graphs with different number of nodes and different parameters $p$. As illustrated in this figure, shortest paths are less likely to overlap in sparse networks. Moreover, in very dense networks, which practically might be less interesting, the shortest path overlap probability decreases as well because most node pairs are connected by one-hop or two-hop paths.

\section{NI with Multiple Snapshots}\label{SIsec:multi-snap-shot}
In this section, we consider the NI problem when multiple snapshots from infection patterns are available. To simplify notation and highlight the main ideas, we consider the single source case with two samples $\by(t_1)$ and $\by(t_2)$ at times $t_1$ and $t_2$, respectively. All arguments can be extended to a more general setup as well.

\noindent
Recall that $\VI^{t}$ denotes the set of infected nodes at time $t$. Let $\EI^{t}=\{(i,j)|(i,j)\in E, \{i,j\}\subset \VI^t\}$ represent edges among infected nodes in the network. The infection subgraph $\GI^{t}=(\VI^t,\EI^t)$ is connected if there is no infection recovery and the underlying diffusion is according to a dynamic process. We define an infection contraction operator $g(.)$ as follows:

\begin{align}\label{eq:contraction-operator}
 g(v) = \begin{cases}
        v  & v\in V\setminus \VI^t\\
        x & o.w.
        \end{cases}
\end{align}
\noindent
where $x\notin V$. In other words, $g(.)$ maps all infected nodes to a new node $x$, while it maps all other nodes to themselves.

\begin{definition}[Infusion Contraction Graph]\label{def:contraction-graph}
Suppose $G=(V,E,W)$ is a weighted graph whose infected subgraph at time $t$ is represented as $\GI^t=(\VI^t,\EI^t,\WI^t)$. An infusion contraction graph $\Gc^t=(\Vc^t,\Ec^t,\Wc^t)$ is defined as follows:
\begin{itemize}
\item $(i,j)\in \Ec^t$ for $i,j\neq x$ iff $(i,j)\in E$. In this case, $w_c^t(i,j)=w(i,j)$.
\item $(i,x)\in \Ec^t$ for $i\neq x$ iff there exists $j\in \VI^t$ such that $(i,j)\in E$. In this case, $w_c^t(i,x)=\sum_{\substack{
j\in\VI^t\\
(i,j)\in E}} w(i,j)$.
\end{itemize}
\end{definition}
\noindent
Intuitively, the infusion contraction graph considers the infected subgraph as one node and adjusts weights of un-infected nodes connected to the infected ones accordingly.
\noindent
Now we consider the source inference problem when two snapshots at times $t_1$ and $t_2$ are given. Recall that $\VI^{t_1}$ and $\VI^{t_2}$ denote the set of infected nodes at times $t_1$ and $t_2$, respectively. Without loss of generality, we assume $0<t_1<t_2$. Using the probability chain-rule, we can re-write the likelihood scores of Optimization \eqref{opt:ml-setup} as,

\begin{align}\label{eq:multi-snap-shot-prob}
& Pr\big(\by(t_1),\by(t_2)|\cS=\{s\}\big)\\
& = \underbrace{Pr\big(\by(t_1)|\cS=\{s\}\big)}_\text{term I} \underbrace{Pr\big(\by(t_2)|\by(t_1),\cS=\{s\}\big)}_\text{term II}.\nonumber
\end{align}
\noindent
Term (I) is the likelihood score of the single source NI Optimization \eqref{opt:NI-ML}. We consider different possibilities for Term II as follows:

\begin{figure}[t]
  \centering
      \includegraphics[height=6cm]{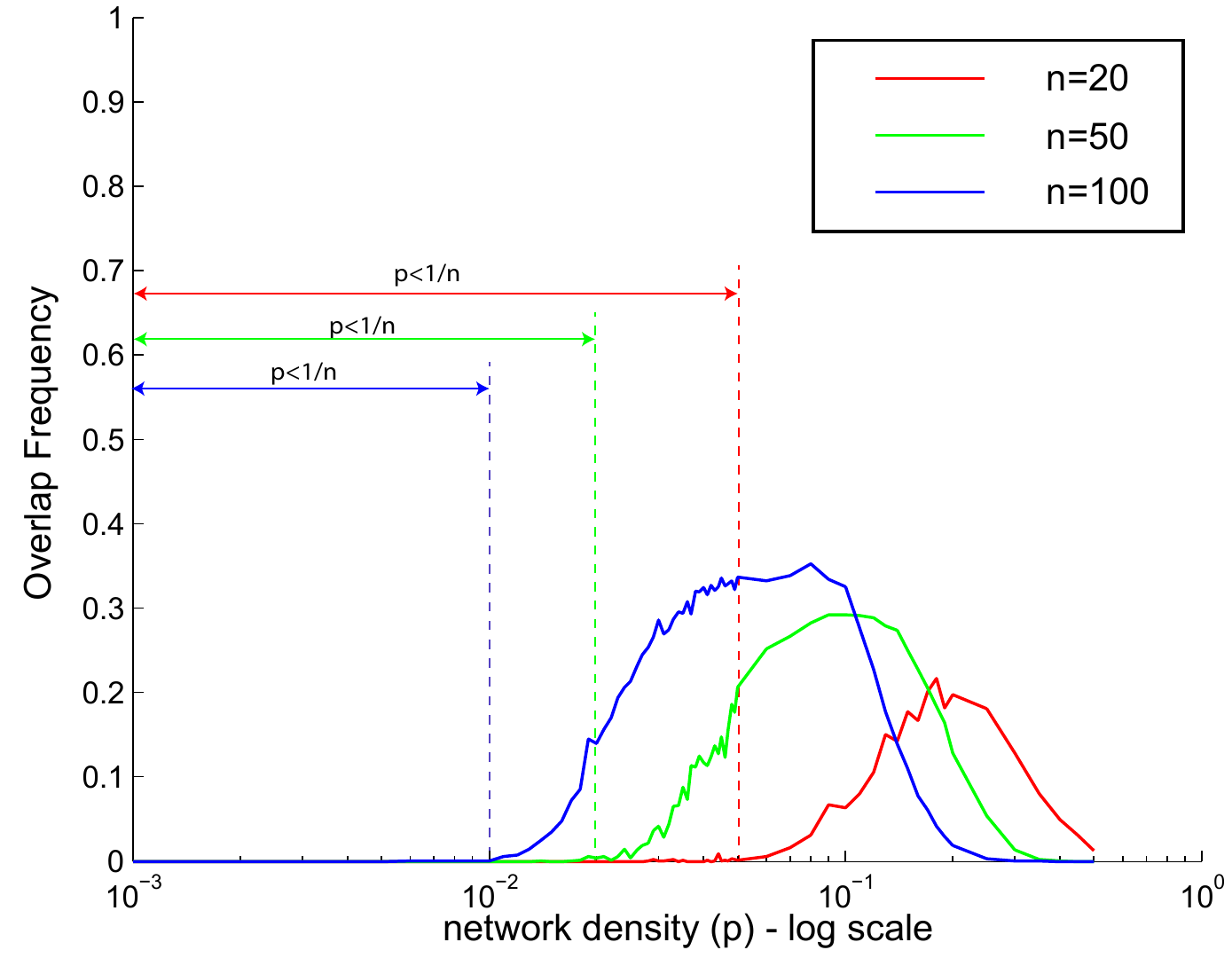}
  \caption{The frequency of having overlapping shortest paths between two randomly selected nodes over an Erd\"os-R\'enyi graph with parameter $p$. In sparse graph ($p\preceq \frac{1}{n}$), the overlap frequency is small. Experiments are repeated $20,000$ times for each case.}
  \label{fig:erdos-path-overlap}
\end{figure}

\begin{itemize}
\item If $y_j(t_1)=1$ and $y_j(t_2)=1$, $Pr(y_j(t_2)|\by(t_1),\cS=\{s\})=1$, because if a node gets infected, it remains infected (there is no recovery). Thus, if $y_j(t_1)=1$ and $y_j(t_2)=0$, $Pr(y_j(t_2)|\by(t_1),\cS=\{s\})=0$.

\item Now we consider the case $y_j(t_1)=0$. Let $\Gc^{t_1}$ be the infusion contraction graph of Definition \ref{def:contraction-graph}. Suppose that all infected nodes at time $t_1$ are mapped to the node $x$. The second term can be approximated as follows:

\begin{align}\label{eq:conditional-diffusion}
Pr(y_j(t_2)=1|\by(t_1),\cS=\{s\})\approx p_{x,j}(t_2-t_1),
\end{align}

where $p_{x,j}(.)$ is the path-based network diffusion kernel over the graph $\Gc^t$. In other words, all infected nodes at time $t_1$ can be viewed as a single source $x$ in the infusion contraction graph. Note that this approximation is tight when the underlying edge holding time distribution is an exponential distribution which has the memory-less property. Moreover, note that to compute the diffusion kernel in this case, we use the infusion contraction graph because infected nodes are not incoherent, and therefore, the multi-source diffusion kernel of Definition \ref{def:kernel-multi-source} cannot be used.
\end{itemize}
\noindent
Under approximation \eqref{eq:conditional-diffusion}, the second term of \eqref{eq:multi-snap-shot-prob} leads to a similar expression for all source candidates, and therefore, the optimization is simplified to a single snapshot one. In practice, one may compute average source likelihood scores using all snapshots to decrease the variance of the source likelihood scores.
\section{NI Over Synthetic Networks}\label{SIsec:results-sim}
In this section, we compare the performance of proposed NI algorithms with other source inference methods over four different synthetic network structures. In our simulations, we assume that there exists a single source in the network, and the underlying diffusion is according to the SI model of Definition \ref{def:SI-model}. In the SI model, edge holding time variables $\cT_{(i,j)}$ are i.i.d. having an exponential distribution with parameter $\lambda=1$. Note that to generate simulated diffusion patterns, we do not use our path-based network diffusion kernel to have a fair performance comparison of our methods with the one of other source inference techniques. We use the following methods in our performance assessment:

\begin{itemize}
\item {\bf Distance Centrality:} This method infers the source node with the minimum shortest path distance from all infected nodes. Suppose $G$ is the underlying network. The distance centrality of node $i$ corresponding to infected nodes at time $t$ is defined as follows:

    \begin{align}
     D_t(i,G)\triangleq \sum_{j\in\VI^t} d(i,j),
    \end{align}

where $d(i,j)$ represents the length of the shortest path between nodes $i$ and $j$. A source node is inferred using the following optimization:
\begin{align}
s=\arg\min_{i\in\VI^t} D_t(i,G).
 \end{align}
If there is no path between two nodes $i$ and $j$, $d(i,j)=\infty$. This makes the distance centrality measure sensitive to noise specially in real world applications. To avoid this issue, for disconnected nodes $i$ and $j$, we assign $d(i,j)=M$, where $M$ is a large number compared to shortest path distances in the network. In our simulations, we set $M$ as 5 times larger than the network diagonal (i.e., the longest shortest path in the network).
\item {\bf Degree Centrality:} This methods infers the source node with highest direct connections to other infected nodes. Degree centrality of node $i$ corresponding to infected nodes $\VI^t$ is defined as follows:
    \begin{align}
     C_t(i,G)\triangleq \sum_{j\in\VI^t} G(i,j).
    \end{align}
 Note that unlike the distance centrality method which considers both direct and indirect interactions among infected nodes, degree centrality only considers direct interactions. To infer the source node using the degree centrality approach, one needs to solve the following optimization:
\begin{align}
s=\arg\min_{i\in\VI^t} C_t(i,G).
 \end{align}
\item {\bf Network Infusion:} We use NI methods based on maximum likelihood (denoted as NI-ML) described in Algorithm \ref{alg:NI-ML-single}, and minimum error (denoted as NI-ME) described in Algorithm \ref{alg:NI-ME-single}. To have a fair comparison with other methods, we assume that the observation time parameter $t$ is unknown and is estimated using the techniques presented in Section \ref{subsec:parameters}.
\end{itemize}

\noindent
In our simulations, we use four types of input networks:
\begin{itemize}
\item {\bf Erd\"os-R\'enyi graphs:} In this case, $G$ is a symmetric random graph where, $Pr[G(i,j)=1]=p$. Networks have $250$ nodes. In our simulations, we use $p=0.01$.

\item {\bf Power law graphs:} We construct $G$ as follows \cite{power-law2001random}; we start with a random subgraph with 5 nodes. At each iteration, a node is added to the network connecting to $\theta$ existent nodes with probabilities proportional to their degrees. This process is repeated till the number of nodes in the network is equal to $n=250$. In our simulations, we use $\theta=2$ which results in networks with the average density approximately $0.01$.
\item {\bf Grid networks:} In this case, $G$ is an undirected square grid network with $250$ nodes. We assume that the source node is located at the center of the grid to avoid boundary effects.

\item {\bf Asymmetric grid networks:} In this case, $G$ is an undirected graph with $250$ nodes. It has $6$ branches connected to the central node, three branches on the right with heavier connectivity among their nodes, and three branches on the left with sparse connectivity. Figure \ref{fig:asym-grid}-a shows an example of such networks with fewer number of nodes. This is an adversarial example to highlight cases where NI methods fail to converge to the optimal value.
\end{itemize}
In multi-path NI methods, we consider top $10$ independent shortest paths to form the $k$-path network diffusion kernel (i.e., $k=10$). However, in the grid network, to enhance the computation time, we consider $k=2$ because different nodes have at most $2$ independent shortest paths connected to the source node. Parameter $t$ is the time at which we observe the infection spread, and it determines the fraction of infected nodes in the network. If $t$ is very large compared to the graph diameter, almost all nodes in the network become infected. On the other hand, for very small values of $t$, the source inference problem becomes trivial. In our simulations, we consider the cases where the number of infected nodes in the network is less than $75\%$ of the total number of nodes, and greater than at least $10$ nodes. We learn this parameter using the observed infection pattern according to techniques introduced in Section \ref{subsec:parameters}.

For evaluation, we sort infected nodes as source candidates according to the score obtained by different methods. High performing methods should assign the highest scores to the source node. The source node should appear on the top of the inferred source candidates. Ideally, if a method assigns the highest score to the source node, the rank of the prediction is one. We use the rank of true sources averaged over different runs of simulations. More formally, suppose $r(\cM,s)$ is the rank of the source node $s$ inferred by using the method $\cM$. In an exact prediction, $r(\cM,s)=1$, while an average rank of a method based on random guessing is $r(\cM,s)=|\VI^t|/2$. If $r(\cM,s)$ is close to one, it means that the true source is among top predictions of the method $\cM$. In each case, we run simulations 100 times.

Figure \ref{fig:main-sim}-(a-d) compares the performance of different source inference methods over both Erd\"os-R\'enyi and power law networks, and in different ranges of the parameter $t$. In both network models and in all diffusion rates, NI Algorithms based on maximum likelihood (NI-ML) and minimum error (NI-ME) outperform other methods. Panels (a) and (c) illustrate the performance when only one sample from the infection pattern at time $t$ is available, for Erd\"os-R\'enyi and power law networks, respectively. Panels (b) and (d) illustrate the performance of different methods when five independent samples from the infection pattern at time $t$ are given, illustrating the mean-field optimality of NI methods according to Theorems \ref{thm:optimality-NI-ML} and \ref{thm:optimality-NI-ME}. Because the underlying networks are sparse, according to Proposition \ref{prop:erdos-overlap}, the performance of maximum likelihood and minimum error NI methods, both shortest-path and multi-path versions, are close to each other in both network models. Notably, unlike NI methods, the performance of other source inference methods such as distance centrality and degree centrality does not tend to converge to the optimal value even for higher sample sizes.

Figure \ref{fig:main-sim}-(e,f) compares the performance of different source inference methods over grid networks in different ranges of the parameter $t$. In the case of having a single sample from the infection pattern, in small ranges of the parameter $t$, when the fraction of infected nodes is less than approximately $n/4$, NI methods and distance centrality have approximately the same performance, significantly outperforming the degree centrality method. In higher diffusion rates, distance centrality outperforms NI-ML method and NI-ML method outperforms NI-ME method. Again in this range, the performance of the degree centrality method is significantly worst than other methods. Having five independent samples from node infection patterns, the performance of NI methods and distance centrality converges to the optimal one.

Unlike tree and sparse Erd\"os-R\'enyi networks, there are multiple overlapping paths among nodes in the grid structures. However, as we illustrate in Figure \ref{fig:main-sim}-(e,f), the performance of NI methods converge to the optimal value, similarly to the case of sparse graphs. The main reason is that grid structures are symmetric and even though paths among nodes overlap significantly, considering shortest paths among nodes approximates the true underlying diffusion based on an SI dynamics closely. In order to have an adversarial example where the performance of NI methods do not converge to the optimal one, we design an asymmetric grid structure illustrated in Figure \ref{fig:asym-grid}. The three branches on the right side of the central node have strong connectivity among themselves, forming a grid, while the ones on the left is sparsely connected to each other, forming a tree. Therefore, the spread of infection will be faster on the right side compared to the left side, under a dynamic SI model. However, considering only the shortest path among nodes does not capture this asymmetric structure. Therefore, the performance of shortest path NI methods diverges from the optimal value as diffusion rate increases (see Figure \ref{fig:asym-grid}). In this case, considering more paths among nodes to form the $k$-path network diffusion kernel according to Definition \ref{def:path-diffusion-kernel} improves the performance of NI methods significantly, because higher order paths partially capture the asymmetric diffusion spread in the network. Note that the degree centrality method has the best performance in this case, because the source node has the highest degree in the network by design. If we select another node to be the source node as illustrated in Figure \ref{fig:asym-grid-noisy}-a, the performance of the degree centrality method becomes worst significantly, indicating its sensitivity to the source location. In the setup of Figure \ref{fig:asym-grid-noisy}, multi-path NI methods outperform the one of other methods, although their performance do not converge to the optimal value.

\begin{figure*}[t]
  \centering
      \includegraphics[width=14cm]{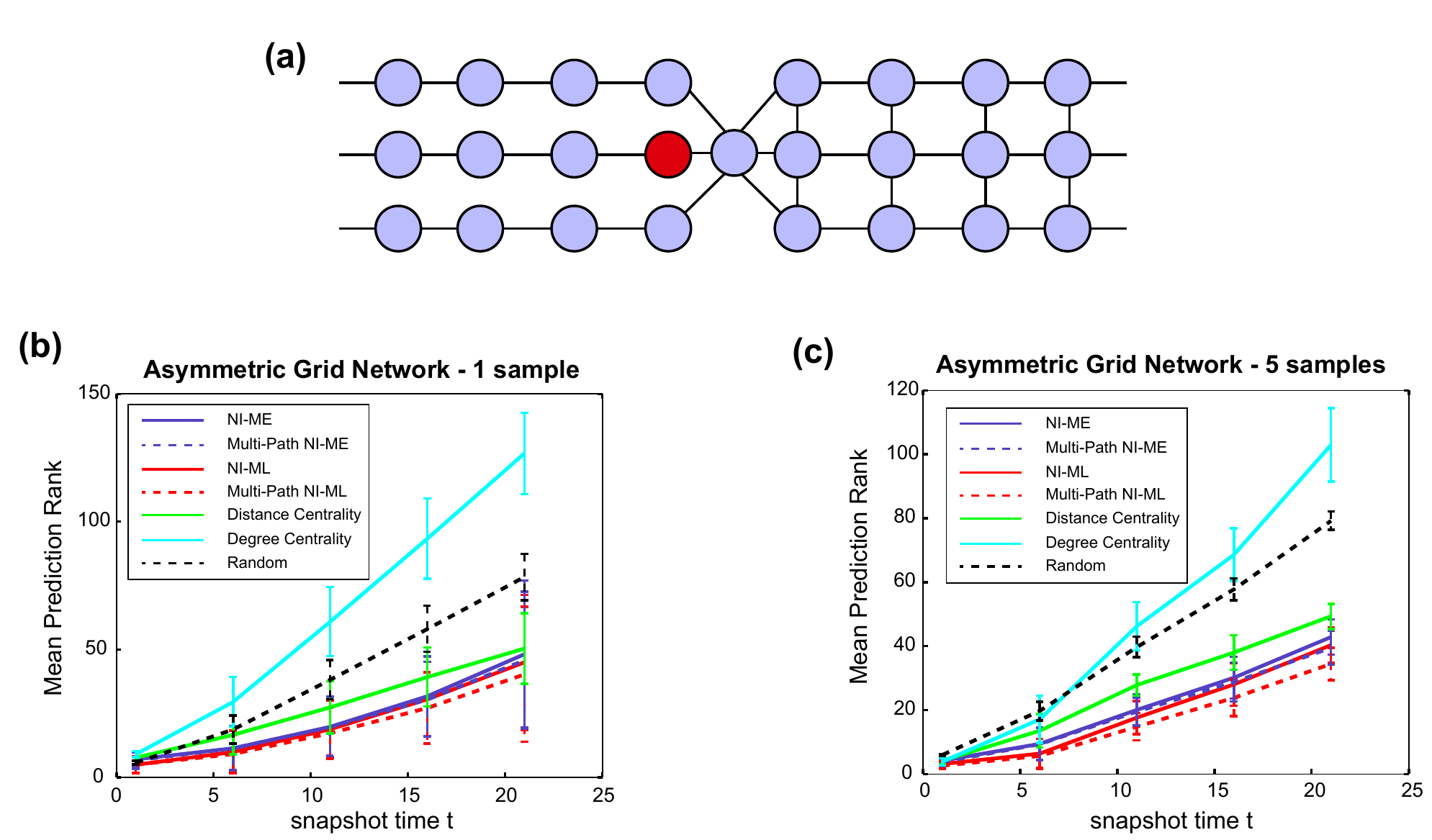}
  \caption{(a) An asymmetric grid network. Performance of source inference methods over asymmetric grid networks with $250$ nodes with (a) one and (b) five independent samples. Node with red color is the source node. Error bars represent $\pm$ standard deviations of empirical results for each case. Experiments have been repeated 100 times so that error margins are small.}
  \label{fig:asym-grid-noisy}
\end{figure*}
\section{Inference of News Sources in a Social News Network}\label{SIsec:digg}
In this section, we evaluate the performance of NI and other source inference techniques in identifying news sources over the social news aggregator Digg (http://digg.com). We also propose an integrative source inference framework which combines network infusion methods (NI-ME and NI-ML) with the distance centrality method. In our integrative framework, we compute the average rank of each node to be the source node according to different source inference methods. We exclude the degree centrality method from our integration owing to its low performance over synthetic and real networks.

Digg allows its users to submit and rate different news links. Highly rated news links are {\it promoted} to the front page of Digg. Digg also allows its users to become friends with other users and follow their activities over the Digg network. Digg's friendship network is asymmetric, i.e., user A can be a follower (friend) of user B but not vice versa. Reference \cite{digg} have collected voting activities and friendship connections of Digg's users over a period of a month in 2009 for $3,553$ promoted stories. We use this data to form a friendship network of active Digg users with $24,219$ nodes and more than $350 K$ connections. We consider users as active if they have voted for at least 10 stories in this time period. Figure \ref{fig:framework} demonstrates a small part of the Digg friendship network.

This application provides an interesting real data framework to assess the performance and robustness of different source inference methods because the true source (i.e., the user who started the news) are known for different stories and also the underlying diffusion processes are based on real dynamics over the Digg friendship network. Moreover, not all of the voting pattern is derived by the source users and there are disconnected voting activities over the Digg friendship network. Thus, performance assessment of different source inference methods in this application can provide a measure of robustness of different methods under real-world circumstances.

Real dynamics over the Digg social news network is in fact significantly different from the standard SI diffusion model. Suppose for a given news propagation, a node $i$ is infected at time $\tau_i$. We compute a realization of random variable $\cT_{(j,i)}$ according to the SI model \ref{def:SI-model}. Over all considered news propagations over the Digg network, around $65\%$ of non-source nodes do not have any infected neighbors in the network at the time of being infected, which violates the most basic conditional independency assumption of the SI model. This can be partially owe to existence of latent sources of information propagation, and/or noisy and spurious network connections. Furthermore, the empirical distribution of remaining news propagation times over edges of the Digg social news network cannot be approximated closely by a single distribution of a homogenous SI diffusion model, even by fitting a general Weibull distribution to the observed data. In Figure \ref{fig:SIR-digg}, we compute the best fit of Weibull distribution (which contains exponential and Rayleigh distributions) to the observed empirical distribution. Parameters of the fitted distribution are $\lambda\approx 3.13$ and $k\approx 0.51$. As it is illustrated in this figure, fitted and empirical distributions differ significantly from each other.

\begin{figure}[t]
  \centering
      \includegraphics[height=5cm]{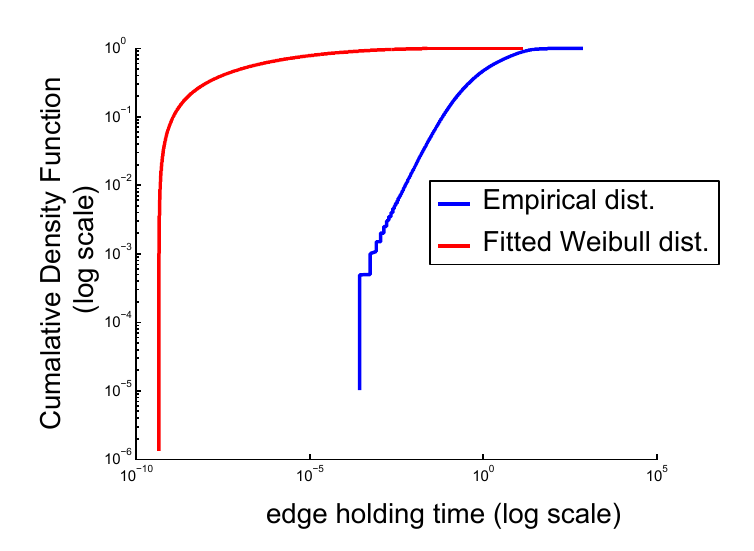}
  \caption{We consider news spread over the Digg social news networks for over 3,500 news stories. In approximately $65\%$ of cases, nodes that have received the news at a time $t$ do not have any neighbors who had already received the news by that time, violating the basic conditional independency assumption of the SI model. Furthermore, the empirical distribution of remaining news propagation times over network edges cannot be approximated closely by a single distribution of a homogenous SI diffusion model, even by fitting a general Weibull distribution to the observed data (Appendix \ref{SIsec:digg}).}
  \label{fig:SIR-digg}
\end{figure}
\section{proofs}\label{sec:proofs}
In this section, we present proofs of the main results of the paper.

\subsection{Proof of Proposition \ref{prop:complexity-diffusion-kernel}}\label{subsec:proof-kernel-complexity}
To compute the $k$-path network diffusion kernel, we need to compute $k$-independent shortest paths among nodes. Note that ties among paths with the same length is broken randomly as explained in Section \ref{subsec:kernel}. Computation of these paths among one node and all other nodes using the Dijkstra's algorithm costs $\cO(k|E|+kn\log(n))$. Thus, computational complexity of forming the entire kernel has complexity $\cO(k|E|n+kn^2\log(n))$.

\subsection{Proof of Proposition \ref{prop:complexity-NI}}\label{subsec:proof-NI-complexity}
Computation of the $k$-path network diffusion kernel for one node has complexity $\cO(k|E|+kn\log(n))$, according to Proposition \ref{prop:complexity-diffusion-kernel}. We need to compute the kernel for $\VI^t$ nodes. Moreover, Optimizations \eqref{opt:NI-ML} and \eqref{opt:NI-ME} have complexity $\cO(|\VI^t|n)$. Thus, the total computational complexity is $\cO\big(|\VI^t|(k|E|+kn\log(n))\big)$.

\subsection{Proofs of Theorem \ref{thm:optimality-NI-ML} and Proposition \ref{prop:robustness-NI-ML}}\label{subsec:proof-ML-optimality}

First, we prove the following lemma:

\begin{lemma}\label{lem:log-inequality}
Let $x$, $y$ and $z$ be positive numbers such that $0<x,y<z$. Define
\begin{align}
f(x,y)\triangleq x \log \frac{x}{y}+ (z-x) \log \frac{z-x}{z-y}.
\end{align}

Then, $f(x,y)\geq 0$ where equality holds iff $x=y$.
\end{lemma}

\begin{proof}
We have,
\begin{align}
\frac{\partial f}{\partial y} &= -\frac{x}{y}+\frac{z-x}{z-y},\nonumber\\
\frac{\partial^2 f}{\partial y^2}&=\frac{x}{y^2}+\frac{z-x}{(z-y)^2}>0.\nonumber
\end{align}
\noindent
For any $0<x0<z$, define $g(y)=f(x0,y)$. Above we show that $g(y)$ is convex, and its minimum is zero which occurs at $y=x0$. Since this holds for any $0<x0<z$, this completes the proof.
\end{proof}
\noindent
Recall that $\cL(i,t')$ is the likelihood score of node $i$ using diffusion parameter $t'$:

\begin{align}
\cL(i,t')=\sum_{j\in\VI^t} \log\big(p_{i,j}(t')\big)+\sum_{j\notin\VI^t} \log\big(1-p_{i,j}(t')\big).
\end{align}
\noindent
First, we prove Proposition \ref{prop:robustness-NI-ML}. Let $s$ be the source node and $t$ be the infection observation time ($t'$ can be different than $t$, see explanations of Remark \ref{remark:NI-ML-t-parameter}). Thus, we can write,

\begin{align}
&\bbE[\cL(s,t)]-\bbE[\cL(s,t')]\\
&=  \sum_{j\in V} p_{s,j}(t)\log\big(p_{s,j}(t)\big)+ \big(1-p_{s,j}(t)\big)\log\big(1-p_{s,j}(t)\big)\nonumber\\
& - \sum_{j\in V} p_{s,j}(t)\log\big(p_{s,j}(t')\big)+ \big(1-p_{s,j}(t)\big)\log\big(1-p_{s,j}(t')\big)\nonumber\\
& =  \sum_{j\in V} p_{s,j}(t)\log\frac{p_{s,j}(t)}{p_{s,j}(t')}+ \big(1-p_{s,j}(t)\big)\log \frac{1-p_{s,j}(t)}{1-p_{s,j}(t')}\nonumber\\
& \stackrel{(I)}{\geq} \sum_{j\in V} p_{s,j}(t)\log\frac{\sum_{j\in V} p_{s,j}(t)}{\sum_{j\in V} p_{s,j}(t')}+ \sum_{j\in V} 1-p_{s,j}(t)\log \frac{\sum_{j\in V}1-p_{s,j}(t)}{\sum_{j\in V}1-p_{s,j}(t')}\nonumber\\
& \stackrel{(II)}{\geq} 0. \nonumber
\end{align}
\noindent
Inequality (I) follows from the log-sum inequality. Inequality (II) follows from Lemma \ref{lem:log-inequality}. In particular, according to Lemma \ref{lem:log-inequality}, the equality condition (II) holds iff $\sum_{j\in V} p_{s,j}(t)=\sum_{j\in V} p_{s,j}(t')$ which indicates $t=t'$. This completes the proof of Proposition \ref{prop:robustness-NI-ML}.
\noindent
In the next step, we prove Theorem \ref{thm:optimality-NI-ML}.

\begin{align}
&\bbE[\cL(s,t)]-\bbE[\cL(i,t')] \\
&=  \sum_{j\in V} p_{s,j}(t)\log\big(p_{s,j}(t)\big)+ \big(1-p_{s,j}(t)\big)\log\big(1-p_{s,j}(t)\big)\nonumber\\
& -  \sum_{j\in V} p_{s,j}(t)\log\big(p_{i,j}(t')\big)+ \big(1-p_{s,j}(t)\big)\log\big(1-p_{i,j}(t')\big)\nonumber\\
& =  \sum_{j\in V} p_{s,j}(t)\log\frac{p_{s,j}(t)}{p_{i,j}(t')}+ \big(1-p_{s,j}(t)\big)\log \frac{1-p_{s,j}(t)}{1-p_{i,j}(t')}\nonumber\\
& \stackrel{(III)}{\geq}  \sum_{j\in V} p_{s,j}(t)\log\frac{\sum_{j\in V} p_{s,j}(t)}{\sum_{j\in V} p_{i,j}(t')}+ \sum_{j\in V} 1-p_{s,j}(t)\log \frac{\sum_{j\in V}1-p_{s,j}(t)}{\sum_{j\in V}1-p_{i,j}(t')}\nonumber\\
& \stackrel{(IV)}{\geq} 0. \nonumber
\end{align}
\noindent
Inequality (III) follows from the log-sum inequality. Inequality (IV) follows from Lemma \ref{lem:log-inequality}. This completes the proof of Theorem \ref{thm:optimality-NI-ML}.

\subsection{Proof of Theorem \ref{thm:optimality-NI-ME}}\label{subsec:proof-ME-optimality}
\begin{figure*}[t]
  \centering
      \includegraphics[height=5cm]{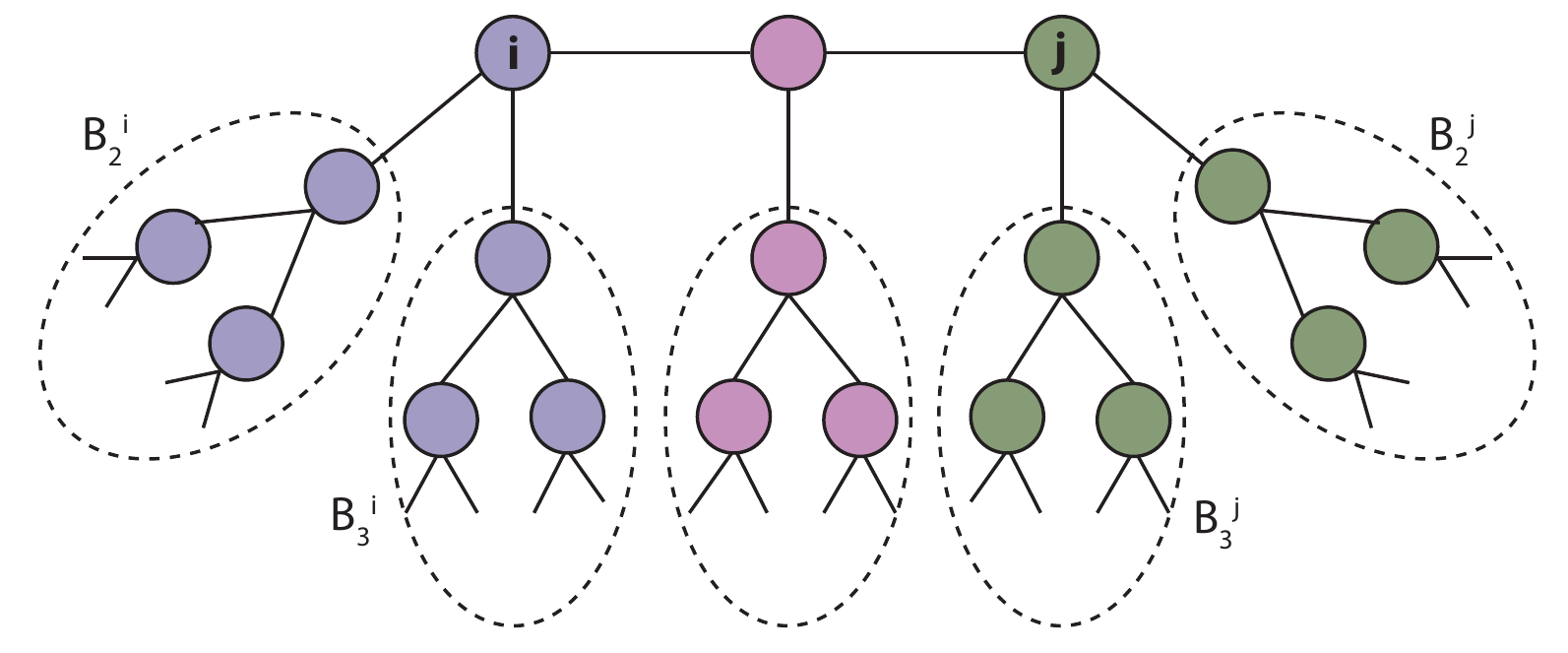}
  \caption{An example of a regular tree used in the proof of Theorem \ref{thm:optimality-NI-ME}.}
  \label{fig:tree}
\end{figure*}
\noindent
To prove Theorem \ref{thm:optimality-NI-ME}, first we prove regular trees are distance symmetric according to the following definition:

\begin{definition}\label{def:distance-symmetric}
A graph $G=(V,E)$ is distance symmetric if for any pair of nodes $i,j\in V$, there exists a graph partition $\{V_1,V_2,V_3\}$ where,
\begin{itemize}
\item $\forall r\in V_1$, $d(i,r)=d(j,r)$. I.e., distances of nodes in $V_1$ from both nodes $i$ and $j$ are the same.
\item There exists a bijective mapping function $\zeta(.)$ between nodes $V_2$ and $V_3$ (i.e., $g:V_2\to V_3$) such that for any $r\in V_2$, $d(i,r)=d(j,\zeta(r))$.
\end{itemize}
\end{definition}
\noindent
In the following Lemma, we show that regular trees are in fact distance symmetric:

\begin{lemma}\label{lem:distance-sym-trees}
Let $G=(V,E)$ be a regular tree where $V$ is countably infinite set of vertices and $E$ is the set of edges. Then, $G$ is distance-symmetric according to Definition \ref{def:distance-symmetric}.
\end{lemma}

\begin{proof}
Consider two distinctive nodes $i$ and $j$. In the following, we construct graph partitions $V_1$, $V_2$ and $V_3$ satisfying conditions of Definition \ref{def:distance-symmetric}. Let degree of nodes in the regular tree be $k$. Thus, there are $k$ branches connected to each nodes $i$ and $j$, denoted by $\{B_1^i,\ldots,B_k^i\}$ and $\{B_1^j,\ldots,B_k^j\}$, respectively. Without loss of generality, assume $i\in B_1^j$ and $j\in B_1^i$. Add $\{B_2^i,\ldots,B_k^i\}$ and $\{B_1^j,\ldots,B_k^j\}$ to the partition sets $V_2$ and $V_3$, respectively. A mapping function $\zeta(.)$ of Definition \ref{def:distance-symmetric} can be constructed between these sets by mapping nodes over branches $B_l^i$ ($l\neq 1$) to nodes over branches $B_l^j$ ($l\neq 1$) in a symmetric way (see Figure \ref{fig:tree}).
\noindent
Now consider the branch connecting nodes $i$ and $j$ in the graph. Let $\{1,2,\ldots,l\}$ be nodes over the shortest path connecting node $i$ to node $j$. Therefore, $d(i,j)=l+1$, the distance between nodes $i$ and $j$.
\begin{itemize}
\item If $l$ is odd, add non-partitioned nodes connected to nodes $\{1,\ldots,l/2\}$ to the partition set $V_2$. Similarly, add remaining nodes connected to nodes $\{l/2+1,\ldots,l\}$ to the partition set $V_3$.

\item If $l$ is odd, add non-partitioned nodes connected to nodes $\{1,\ldots,\lfloor l/2 \rfloor\}$ and $\{\lceil l/2 \rceil+1,\ldots,l\}$ to partition sets $V_2$ and $V_3$, respectively. Non-partitioned nodes connected to the node $\lceil l/2 \rceil$ are assigned to the partition set $V_1$.
\end{itemize}
\noindent
A mapping function $\zeta(.)$ of Definition \ref{def:distance-symmetric} can be constructed between newly added nodes to partition sets $V_2$ and $V_3$ in a symmetric way. Moreover, nodes in the partition set $V_1$ have the same distance from both nodes $i$ and $j$. This completes the proof.
\end{proof}
\noindent
Without loss of generality, suppose node $0$ is the source node and we observe the infection pattern at time $t$ according to the SI diffusion model. Thus, $Pr[y_{j}(t)=1]=p_{0,j}(t)$, defined according to equation \eqref{eq:erlang-kernel}. Suppose we use parameter $t'$ in Network Infusion Algorithm \ref{alg:NI-ME-single}. According to equation \eqref{opt:NI-ME}, we have,

\begin{align}
&\bbE[\cH_{\alpha}(i,t')]\\
&=\sum_{j\in V} (1-\alpha) p_{0,j}(t) \big(1-p_{i,j}(t')\big)+ \alpha \big(1-p_{0,j}(t)\big) p_{i,j}(t')\nonumber\\
&=\sum_{j\in V} p_{0,j}(t)\big(1-p_{i,j}(t')\big)+\alpha \big(p_{i,j}(t')-p_{0,j}(t)\big).\nonumber
\end{align}
\noindent
Thus, we have,

\begin{align}
&\bbE[\cH_{\alpha}(i,t')]-\bbE[\cH_{\alpha}(0,t')]\\
&=\sum_{j\in V} p_{0,j}(t)\big(p_{0,j}(t')-p_{i,j}(t')\big)+\alpha \big(p_{i,j}(t')-p_{0,j}(t')\big)\nonumber\\
&=\sum_{j\in V} \big(p_{0,j}(t')-p_{i,j}(t')\big)(p_{0,j}(t)-\alpha)\nonumber\\
&\overset{\rm (a)}{=}\sum_{j\in V_2} \big(p_{0,j}(t')-p_{i,j}(t')\big)(p_{0,j}(t)-\alpha)\nonumber\\
&+\sum_{j'=g(j)\in V_3}\big(p_{0,j'}(t')-p_{i,j'}(t')\big)(p_{0,j'}(t)-\alpha)\nonumber\\
&\overset{\rm (b)}{=}\sum_{j\in V_2} \big(p_{0,j}(t')-p_{i,j}(t')\big)(p_{0,j}(t)-\alpha)\nonumber\\
&+\sum_{j\in V_2}\big(p_{i,j}(t')-p_{0,j}(t')\big)(p_{i,j}(t)-\alpha)\nonumber\\
&=\sum_{j\in V_2} \big(p_{0,j}(t')-p_{i,j}(t')\big)\big(p_{0,j}(t)-p_{i,j}(t)\big).\nonumber
\end{align}
\noindent
Equality (a) comes from partitioning nodes to sets $V_1$, $V_2$ and $V_3$ according to Definition \ref{def:distance-symmetric}. The terms correspond to nodes in the partition set $V_1$ is equal to zero. Equality (b) comes from the fact that $d(0,j')=d(i,j)$ and $d(0,j)=d(i,j')$. Thus, $p_{0,j'}(.)=p_{i,j}(.)$ and $p_{i,j'}(.)=p_{0,j}(.)$.
\noindent
Therefore, if $t'=t$, we have,

\begin{align}
\bbE[\cH_{\alpha}(i,t')]-\bbE[\cH_{\alpha}(0,t')]=\sum_{j\in V_2} (p_{0,j}(t)-p_{i,j}(t))^2,
\end{align}
\noindent
which is strictly positive if $i\neq 0$.
\noindent
Now we consider the case that $t'\neq t$. Suppose $d_{0,j}<d_{i,j}$. Then, $p_{0,j}(t)>p_{0,j}(t)$ for any value of $t>0$. Therefore, $\big(p_{0,j}(t')-p_{i,j}(t')\big)\big(p_{0,j}(t)-p_{i,j}(t)\big)>0$. The same argument holds if $d_{0,j}>d_{i,j}$. If $d_{0,j}=d_{i,j}$, then $\big(p_{0,j}(t')-p_{i,j}(t')\big)\big(p_{0,j}(t)-p_{i,j}(t)\big)=0$. This completes the proof of Theorem \ref{thm:optimality-NI-ME}.

\subsection{Proof of Theorem \ref{thm:optimality-multi-source}}\label{subsec:proof-multi-source-optimality}

To simplify notation, we prove this Theorem for a specific case where there are three sources in the network ($m=3$). Also, we drop $\eps$ from $d_1^{\eps}$. All arguments can be extended to a general case. Let $\cS=\{0,1,2\}$ be the sources. Suppose at the first step of the Algorithm \ref{alg:multi-source--NI}, we have inferred the source node $0$. We show that at the next step, we have,

\begin{align}
\bbE[\cL_{d_0}(s,t)]\geq \bbE[\cL_{d_0}(i,t)],
\end{align}
\noindent
for $s\in \{1,2\}$ and for all $i\in \VI^t-D(0,d_1)$.
\noindent
Consider a node $i$ in the $d_1$-neighborhood of the source node $1$ (i.e., $i\in D(1,d_1)$. Consider a node $j$ in $d_0$-neighborhood of node $i$ (Figure \ref{fig:multi-proof-fig}). According to Equation \eqref{eq:kernel-multiple-sources}, we have,

\begin{align}\label{eq:prob-apx-proof-multi}
Pr(y_j(t)=1)&\geq p_{1,j}\\
Pr(y_j(t)=1)& \stackrel{(I)}{\leq} \sum_{s=1}^{3} p_{s,j}\nonumber\\
& \stackrel{(II)}{\leq} p_{1,j}+\frac{\eps}{n},\nonumber
\end{align}
\noindent
where Inequality (I) comes from the union bound of probabilities, and Inequality (II) uses incoherent source property of Definition \ref{def:eps-coherent}.
\noindent
Consider a node $i$ in the $d_1$-neighborhood of the source node $1$ (i.e., $i\in D(1,d_1))$. For this node, we have,

\begin{align}
\bbE[\cL_{d_0}(i,t)]&=\sum_{j\in D(i,d_0)} Pr\big(y_j(t)=1\big) \log \big(p_{i,j}(t)\big)\\
&+Pr\big(y_j(t)=0\big) \log \big(1-p_{i,j}(t)\big)\nonumber\\
&\stackrel{(III)}{=}\sum_{j\in D(i,d_0)} (p_{1,j}(t)+\frac{\eps_j}{n}) \log \big(p_{i,j}(t)\big)\nonumber\\
&+(1-p_{1,j}(t)-\frac{\eps_j}{n}) \log \big(1-p_{i,j}(t)\big)\nonumber\\
&=\sum_{j\in D(i,d_0)} p_{1,j}(t) \log \big(p_{i,j}(t)\big)\nonumber\\
&+(1-p_{1,j}(t)) \log \big(1-p_{i,j}(t)\big)\nonumber\\
&+ \underbrace{\sum_{j\in D(i,d_0)} \frac{\eps_j}{n} \log \frac{p_{i,j}(t)}{1-p_{i,j}(t)}}_\text{term IV}\nonumber\\
& \asymp \sum_{j\in D(i,d_0)} p_{1,j}(t) \log \big(p_{i,j}(t)\big)\nonumber\\
&+(1-p_{1,j}(t)) \log \big(1-p_{i,j}(t)\big),\nonumber
\end{align}
\noindent
where Equality (I) comes from Equation \eqref{eq:prob-apx-proof-multi}, and term (IV) goes to zero for sufficiently large $n$ and a fixed $t$.
\noindent
Similarly to the proof of Theorem \ref{thm:optimality-NI-ML}, we have

\begin{align}\label{eq:multi-proof-intersection}
&\sum_{j\in D(i,d_0)\cap D(1,d_0)} p_{1,j}(t) \log \big(p_{1,j}(t)\big)\\
&+(1-p_{1,j}(t)) \log \big(1-p_{1,j}(t)\big)\nonumber\\
&\geq \sum_{j\in D(i,d_0)\cap D(1,d_0)} p_{1,j}(t) \log \big(p_{i,j}(t)\big)\nonumber\\
&+(1-p_{1,j}(t)) \log \big(1-p_{i,j}(t)\big).\nonumber
\end{align}
\noindent
Note that this inequality holds for nodes in the $d_0$-neighborhood of both nodes $i$ and $1$. Now consider a node $j\in D(i,d_0)-D(1,d_0)$, and a node $j'\in D(1,d_0)-D(i,d_0)$ (see Figure \ref{fig:multi-proof-fig}). Owing to the symmetric structure of the network, similarly to Lemma \ref{lem:distance-sym-trees}, there is a one-to-one map among nodes $j$ and $j'$ such that $d(i,j)=d(1,j')$. For such node pairs $j$ and $j'$, we have,

\begin{align}\label{eq:multi-proof-difference}
&p_{1,j'}(t) \log(p_{1,j'}(t))+\big(1-p_{1,j'}(t)\big) \log\big(1-p_{1,j'}(t)\big)\\
&- p_{1,j}(t) \log(p_{i,j}(t))-\big(1-p_{1,j}(t)\big) \log\big(1-p_{i,j}(t)\big)\nonumber\\
&=p_{1,j'}(t) \log(p_{1,j'}(t))+\big(1-p_{1,j'}(t)\big) \log\big(1-p_{1,j'}(t)\big)\nonumber\\
&- p_{1,j}(t) \log(p_{1,j'}(t))-\big(1-p_{1,j}(t)\big) \log\big(1-p_{1,j'}(t)\big)\nonumber\\
&= (p_{1,j'}(t)-p_{1,j}(t))\log \frac{p_{1,j'}(t)}{1-p_{1,j'}(t)}\nonumber\\
& \geq 0,
\end{align}
\noindent
where the inequality comes from the fact that $d(1,j')<d(1,j)$. Thus, we have,

\begin{align}\label{eq:multi-proof-difference2}
&\sum_{j\in D(1,d_0)-D(1,d_0)} p_{1,j}(t) \log \big(p_{1,j}(t)\big)\\
&+(1-p_{1,j}(t)) \log \big(1-p_{1,j}(t)\big)\nonumber\\
&\geq \sum_{j\in D(i,d_0)- D(1,d_0)} p_{1,j}(t) \log \big(p_{i,j}(t)\big)\nonumber\\
&+(1-p_{1,j}(t)) \log \big(1-p_{i,j}(t)\big).\nonumber
\end{align}
\noindent
Combining Inequalities \eqref{eq:multi-proof-intersection} and \eqref{eq:multi-proof-difference2}, we have,

\begin{figure}[t]
  \centering
      \includegraphics[height=6cm]{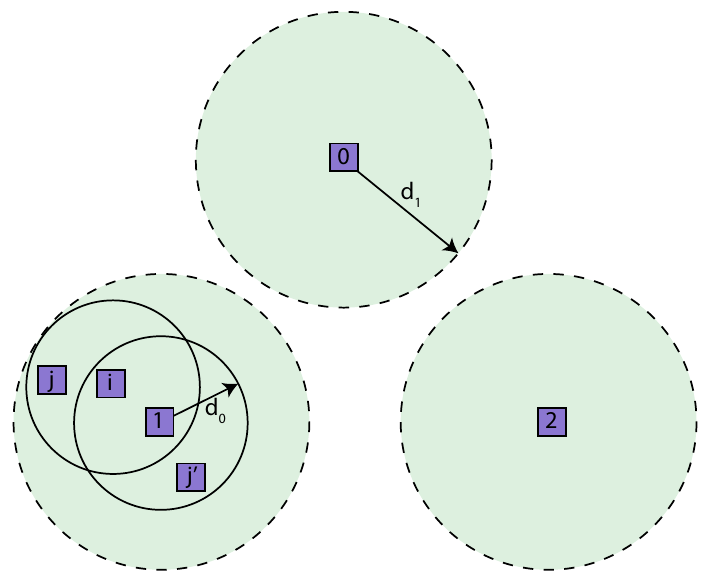}
  \caption{An illustrative figure of the proof of Theorem \ref{thm:optimality-multi-source}.}
  \label{fig:multi-proof-fig}
\end{figure}

\begin{align}
\bbE[\cL_{d_0}(1,t)]\geq \bbE[\cL_{d_0}(i,t)],
\end{align}
\noindent
for any node $i$ in the $d_1$-neighborhood of the source node $1$. The same arguments can be repeated for nodes in the $d_1$-neighborhood of the source node $2$. There are some remaining nodes that are not in the $d_1$-neighborhood of the sources. As the last step of the proof, we show that the probability of having an infected remaining node is small. Consider node $j$ such that $d(j,\cS)>d_1$. according to Equation \eqref{eq:kernel-multiple-sources} and using the probability union bound, we have
\begin{align}
Pr(y_j(t)=1)\leq \frac{\eps}{n}.
\end{align}
\noindent
Let $p_e$ denote the probability of at least one such infected node exists. We have,

\begin{align}
p_e &\leq 1-\big(1-\frac{\eps}{n}\big)^n\\
&\asymp \eps,\nonumber
\end{align}
\noindent
for sufficiently large $n$. This completes the proof of Theorem \ref{alg:multi-source--NI}.

\subsection{Proof of Proposition \ref{prop:complexity-greedy-NI}}\label{subsec:proof-complexity-multi-source}
Computation of the $k$-path network diffusion kernel for infected nodes has computational complexity $\cO\big(|\VI^{t}|(k|E|+kn\log(n))\big)$ according to Proposition \ref{prop:complexity-diffusion-kernel}. Moreover, solving Optimization \eqref{opt:ml-multi-source-greedy} for $m$ iterations costs $\cO(|\VI^{t}|nm)$. Thus, the total computational complexity of Algorithm \ref{alg:multi-source--NI} is $\cO\big(|\VI^{t}|(k|E|+kn\log(n)+mn)\big)$.

\end{appendices}


\end{document}